\documentclass{article}
\usepackage{graphicx}
\usepackage{hyperref}
\usepackage{amsmath,amsfonts,amssymb,amsthm}
\usepackage{cleveref}
\usepackage{enumerate}
\usepackage[sortcites, style=alphabetic,
maxbibnames=99, hyperref]{biblatex}
\usepackage{geometry}
\usepackage{xcolor}
\usepackage{algorithm}
\usepackage{bbm}
\usepackage{thm-restate}
\usepackage{mathabx}
\usepackage{thmtools}
\usepackage{csquotes}
\usepackage{authblk}

\global\long\def\mem{\mathtt{mem}}
\global\long\def\mate{\mathtt{mate}}
\global\long\def\match{\mathtt{match}}
\global\long\def\eval{\mathtt{eval}}
\global\long\def\emd{\mathrm{EMD}}
\global\long\def\tv{\mathrm{TV}}

\global\long\def\polylog{\mathrm{polylog}}
\global\long\def\sizeIn{\delta_{\mathrm{in}}}
\global\long\def\sizeOut{\delta_{\mathrm{out}}}

\global\long\def\qmsndap{QMSNDAP }
\global\long\def\qmf{QMF}

\newtheorem{theorem}{Theorem}
\newtheorem*{theorem*}{Theorem}
\newtheorem{lemma}{Lemma}[section]

\newtheorem{definition}[lemma]{Definition}

\theoremstyle{definition}
\newtheorem{observation}[theorem]{Observation}

\theoremstyle{plain}

\newcommand{\ld}{\left}
\newcommand{\rd}{\right}
\newcommand{\cond}{\;|\;}

\newcommand{\inp}{\texttt{in}}
\newcommand{\out}{\texttt{out}}

\newcommand{\augment}{\texttt{Augment}}
\newcommand{\augmenteligible}{\texttt{AugmentEligible}}
\newcommand{\largematch} {\texttt{LargeMatching}}
\newcommand{\largematchforward}{\texttt{LargeMatchingForward}}
\newcommand{\approxmatch} {\texttt{ApproximateMatching}}
\newcommand{\opt}{\texttt{OPT}}
\newcommand{\alg}{\texttt{ALG}}

\newcommand{\cE}{\mathcal E}
\newcommand{\bbZ}{\mathbb Z}
\newcommand{\bbR}{\mathbb R}
\newcommand{\cP}{\mathcal P}
\newcommand{\cQ}{\mathcal Q}
\newcommand{\cF}{\mathcal F}
\newcommand{\cV}{\mathcal V}
\newcommand{\cA}{\mathcal A}
\newcommand{\cC}{\mathcal C}
\newcommand{\cR}{\mathcal R}
\newcommand{\cT}{\mathcal T}
\newcommand{\cM}{\mathcal M}
\newcommand{\one}{\mathbbm 1}

\newcommand{\wbar}{\bar w}
\newcommand{\Mbar}{\widebar M}
\newcommand{\Gbar}{\widebar G}
\newcommand{\Vbar}{\widebar V}
\newcommand{\Ebar}{\widebar E}
\newcommand{\cbar}{\bar c}
\newcommand{\Otil}{\tilde{O}}

\newcommand{\E}{\textnormal{E}}
\newcommand{\Var}{\textnormal{Var}}

\title{Approximate Earth Mover's Distance in Truly-Subquadratic Time}

\author[1]{Lorenzo Beretta}
\author[2]{Aviad Rubinstein}
\affil[1]{BARC, University of Copenhagen, \texttt{lorenzo2beretta@gmail.com}}
\affil[2]{Stanford University,
\texttt{aviad@cs.stanford.edu}}

\date{}

\begin{document}
\maketitle

\begin{abstract}
We design an additive approximation scheme for estimating the cost of the min-weight bipartite matching problem: given a  bipartite graph with non-negative edge costs and $\varepsilon > 0$, our algorithm estimates the cost of matching all but $O(\varepsilon)$-fraction of the vertices in truly subquadratic time $O(n^{2-\delta(\varepsilon)})$. 

\begin{itemize}
    \item Our algorithm  has a natural interpretation for computing the Earth Mover's Distance (EMD), up to a $\varepsilon$-additive approximation. Notably, we make no assumptions about the underlying metric (more generally, the costs do not have to satisfy triangle inequality). Note that compared to the size of the instance (an arbitrary $n \times n$ cost matrix), our algorithm runs in {\em sublinear} time.
    \item Our algorithm can approximate a slightly more general problem: max-cardinality bipartite matching with a knapsack constraint, where the goal is to maximize the number of vertices that can be matched up to a total cost $B$.
\end{itemize}
\end{abstract}
\section{Introduction}

{\em Earth Mover's Distance} (EMD - sometimes also Optimal Transport, Wasserstein-$1$ Distance or Kantorovich–Rubinstein Distance) is perhaps the most important and natural measure of similarity between probability distributions over elements of a metric space \cite{villani2009optimal, santambrogio2015optimal, peyre2019computational}. Formally, given two probability distributions $\mu$ and $\nu$ over a metric space $(\cM, d)$ their EMD is defined as

\begin{align}
\label{eq:emd objective}
    \emd(\mu, \nu) = \min \ld\{ \E_{(x, y) \sim \zeta}[d(x, y)] \;\Big|\; \zeta \text{ is a coupling\footnotemark of $\mu$ and $\nu$} \rd\}.
\end{align}
\footnotetext{A distribution $\zeta$ over $\cM^2$ is a coupling of $\mu$ and $\nu$ if 
$\mu(x) = \int_\cM \zeta(x, y) \,dy$ and $\nu(y) = \int_\cM \zeta(x, y) \,dx$.}

When $\mu$ and $\nu$ are discrete distributions with 
support size $n$ (perhaps after a discretization preprocessing), a straightforward algorithm for estimating their EMD is to sample $\Theta(n)$ elements from each, compute all $\Theta(n^2)$ pairwise distances, and then compute a bipartite min-weight perfect matching. 
This algorithm clearly takes at least $\Omega(n^2)$ time (even ignoring the computation of the matching), and incurs a small additive error due to the sampling.

Our main result is an asymptotically faster algorithm for estimating the EMD:

\begin{restatable}[Main Theorem]{theorem}{MainThmIntroEMD}
\label{thm:main theorem intro EMD}
Suppose we have sample access to two distributions $\mu,\nu$ over metric space $(\cM, d)$ satisfying $d(\cdot, \cdot) \in [0,1]$ and query access to $d$. Suppose further that $\mu,\nu$ have support size at most $n$. 

For each constant $\gamma > 0$ there exists a constant $\varepsilon > 0$ and an algorithm running in time $O(n^{2 - \varepsilon})$ that outputs $\widehat{\emd}$ such that
\[
 \widehat \emd \in [\emd(\mu, \nu) \pm \gamma].
\]
Moreover, such algorithm takes $\tilde O(n)$ samples from $\mu$ and $\nu$.
\end{restatable}
Notably, our algorithm makes no assumption about the structure of the underlying metric. In fact, it can be an arbitrary non-negative cost function, i.e.~we do not even assume triangle inequality. 

\paragraph{Beyond bounded support size.}
Support size is a brittle matter; indeed two distributions that are arbitrarily close in total variation (TV) distance (or EMD) can have completely different support size. Moreover, for continuous distributions, the notion of support size is clearly inappropriate and yet we would like to compute their EMD through sampling.
To obviate this issue, \Cref{thm:main theorem intro EMD generalized} generalize
\Cref{thm:main theorem intro EMD} to distributions that are \emph{close} in EMD to some distributions with support size $n$.

\begin{restatable}{corollary}{MainThmIntroEMDGen}
\label{thm:main theorem intro EMD generalized}
Suppose we have sample access to two distributions $\mu,\nu$ over metric space $(\cM, d_\cM)$ satisfying $d(\cdot, \cdot) \in [0,1]$ and query access to $d$. Suppose further that there exist $\mu',\nu'$ with support size $n$ such that $\emd(\mu, \mu'), \emd(\nu, \nu') \leq \xi$, for some $\xi > 0$. 

For each constant $\gamma > 0$ there exists a constant $\varepsilon > 0$ and an algorithm running in time $O(n^{2 - \varepsilon})$ that outputs $\widehat{\emd}$ such that
\[
 \widehat \emd \in [\emd(\mu, \nu) \pm (4\xi + \gamma)].
\]
Moreover, such algorithm takes $\tilde O(n)$ samples from $\mu$ and $\nu$.
\end{restatable}

For continuous $\mu$, requiring that $\mu$ is close in EMD to a distribution with bounded support size is equivalent to saying that $\mu$ can be discretized effectively for $\emd$ computation. Thus, such assumption is natural while computing $\emd$ between continuous distribution through discretization.

We stress that the algorithm in \Cref{thm:main theorem intro EMD generalized} does not assume knowledge of $\mu'$ (nor $\nu'$) beyond its support size $n$. Indeed, the empirical distribution over $\tilde O(n)$ samples from $\mu$ (resp. $\nu$) makes a good approximation in $\emd$. 
Finally, the sample complexity in \Cref{thm:main theorem intro EMD} and \Cref{thm:main theorem intro EMD generalized} is optimal, up to $\polylog(n)$ factors. Indeed, Theorem 1 in \cite{valiant2010clt} implies a lower bound of $\tilde \Omega(n)$ on the sample complexity of testing EMD closeness\footnote{While deriving the lower bound from \cite{valiant2010clt} takes some work, Remark 5.13 in \cite{canonne2020survey} explicitly states a $\Omega(n / \log n)$ lower bound for TV closeness testing.}.

\paragraph{Matching with knapsack constraint.}
Applying our main algorithm to a graph-theory setting, we give an approximation scheme for a knapsack bipartite matching problem, where our goal is to estimate the number of vertices that can be matched subject to a total budget constraint. 

\begin{restatable}[Main theorem, graph interpretation]{theorem}{MainThmIntroGraph}
\label{thm:main theorem intro graph}
For each constant $\gamma > 0$, there exists a constant $\varepsilon > 0$, and an algorithm running  in time $O(n^{2 - \varepsilon})$ with the following guarantees. The algorithm takes as input a budget $B$, and query access to the edge-cost matrix of an  undirected, bipartite graph $G$ over $n$ vertices. The algorithm returns an estimate $\widehat{M}$ that is within $\pm \gamma n$ of the size of the maximum matching in $G$ with total cost at most $B$.
\end{restatable}

\subsection{Relaed Work}

Computing EMD is an important problem in machine learning \cite{peyre2019computational} with some exemplary applications in computer vision \cite{rubner2000earth, arjovsky2017wasserstein, solomon2015convolutional} and natural language processing \cite{kusner2015word, yurochkin2019hierarchical}. See \cite{peyre2019computational} for a comprehensive overview. 

\paragraph{Exact solution.}
Computing EMD between two sets of $n$ points boils down to computing the minimum cost of a perfect matching on a bipartite graph, a problem with a $70$-years history \cite{kuhn1955hungarian}. Min-weight bipartite perfect matching can be cast as a min-cost flow (MCF) instance and to date we can solve it in $n^{2 + o(1)}$ time (namely, near-linear in the size of the distance matrix) \cite{chen2022maximum}.
Apparently, any exact algorithm requires inspecting the entire distance matrix, thus $\Theta(n^2)$ time is the best we can hope for. 
In addition, even in $d$-dimensional Euclidean space, where the input has size $d \cdot n \ll n^2$, no $O(n^{2-\varepsilon})$ algorithm exists\footnote{The lower bound in \cite{rohatgi2019conditional} holds in dimension $d = 2^{\Omega(\log^* n)}$.}, unless SETH is false \cite{rohatgi2019conditional}.

\paragraph{Multiplicative approximation.}
A significant body of work has investigated multiplicative approximation of EMD
\cite{charikar2002similarity, indyk2003fast, chen2022new, agarwal2022deterministic, agarwal2014approximation, andoni2023sub, andoni2008earth, andoni2009efficient, andoni2014parallel}, where the most commonly studied setting is the Euclidean space (or, more generally, $\ell_p$). 
If the dimension is constant we have near-linear time approximation schemes \cite{sharathkumar2012near, andoni2014parallel, agarwal2022deterministic, fox2022deterministic}, whereas the high-dimensional case is more challenging.
Only recently \cite{andoni2023sub} broke the $O(n^2)$ barrier for $(1+\varepsilon)$-approximation of EMD, building on \cite{har2013euclidean}.

The landscape is much less interesting for general metrics. Indeed, a straightforward counterexample from \cite{buadoiu2005facility} shows that any $O(1)$-approximation requires $\Omega(n^2)$ queries to the distance matrix. This suggests that for general metrics we should content ourselves with a additive approximation. 

\paragraph{Additive approximation.}
Additive approximation for EMD has been extensively studied by optimization and machine learning communities \cite{altschuler2017near, blanchet2018towards, dvurechensky2018computational, luo2023improved, cuturi2013sinkhorn, le2021robust, pham2020unbalanced}.

An extremely popular algorithm to solve optimal transport in practice is Sinkhorn algorithm \cite{cuturi2013sinkhorn} (see \cite{le2021robust, pham2020unbalanced} for recent work). Sinkhorn distance SNK is defined by adding an entropy regularization term $-\eta \cdot H(\zeta)$ to the EMD objective in \Cref{eq:emd objective}. Approximating SNK via Sinkhorn algorithm provably yields a $\varepsilon r$-additive approximation to EMD and takes $O_\varepsilon(n^2)$ time, where $r$ is the dataset diameter \cite{altschuler2017near}.

Graph-theoretic approaches also led to $\varepsilon r$-additive approximations \cite{lahn2019graph} in $O_\varepsilon(n^2)$ time. Notice that even though all previous approximation algorithms have roughly the same complexity as the MCF-based exact solution they are backed by experiments showing their practicality, whereas exact algorithms for EMD are still largely impractical for very large graphs.

\paragraph{Breaking the $O(n^2)$ barrier for general metrics.}
As mentioned above, \cite{andoni2023sub} was the first work to break the quadratic barrier for approximate EMD. Indeed, they show a $(1+\varepsilon)$-multiplicative approximation algorithm for EMD on Euclidean space running in $n^{2-\Omega_\varepsilon(1)}$ time. 
Matching such result on general metrics is impossible, since no $O(1)$-multiplicative approximation can be achieved in $o(n^2)$ time \cite{buadoiu2005facility}. 
A natural way to bypass the lower bound in \cite{buadoiu2005facility} is to consider additive approximation. 
However, no $\varepsilon$-additive approximation algorithm for EMD on general metrics faster than $O_\varepsilon(n^2)$ barrier was known prior to this work.
\Cref{thm:main theorem intro EMD} gives the first $\varepsilon$-additive approximation to EMD for general metrics running in $n^{2 - \Omega_\varepsilon(1)}$ time, thus breaking the quadratic barrier for general metrics. 

We stress that, despite \cite{andoni2023sub} and this work both prove similar results, they use a completely different set of techniques. Indeed, in \cite{andoni2023sub} they approximate the complete bipartite weighted graph induced by Euclidean distances with a $(1+\varepsilon)$-multiplicative spanner of size $n^{2-\Omega_\varepsilon(1)}$. 
Their spanner construction is based on LSH and so it hinges on the Euclidean structure. Then, they run a near-linear time MCF solver \cite{chen2022maximum} to solve the matching problem on the metric induced by the spanner. In this work, instead, we build on sublinear algorithms for max-cardinality matching \cite{behnezhad2022time, behnezhad2023beating, behnezhad2023sublinear, BKS23a, bhattacharya2023dynamic} and do not leverage any metric property, not even triangle inequality. \Cref{sec:tech overview} contains a detailed explanation of our techniques.

It is worth to notice that since \cite{andoni2023sub} operates over $d$-dimensional Euclidean space the input representation takes $d\cdot n$ space, and so it \emph{does not} run in sublinear time. On the contrary, our algorithm assumes query access to the distance matrix and runs in sublinear time.

\paragraph{Sublinear algorithms.}
Most previous work in sublinear models of computation focuses on streaming Euclidean EMD \cite{chen2022new, andoni2009efficient, andoni2008earth, backurs2020scalable, indyk2004algorithms, charikar2002similarity}, where the latest work \cite{chen2022new} achieves $\tilde O(\log n)$-approximation in polylogarithmic space.
Some other work \cite{ba2011sublinear} addresses the sample complexity of testing $\emd$ on low-dimensional $\ell_1$.

In this work we focus on a different access model: we do not make any assumption on the ground metric and we assume query access to the distance matrix. This model is natural whenever the underlying metric is expensive to evaluate. For example, in  \cite{abeywickrama2021optimizing} they consider EMD over a shortest-path ground metric and experiment with heuristics to avoid computing all-pair distances, which would be prohibitively expensive.

\paragraph{Comparison with MST.}
Minimum Spanning Tree (MST) and EMD are two of the most studied optimization problems in metric spaces.
It is interesting to observe a separation between the sublinear-time complexity of MST and EMD for general metrics. Indeed, \cite{czumaj2009estimating} shows a $\tilde O_\varepsilon(n)$ time algorithm approximating the \emph{cost} of MST up to a factor $1+\varepsilon$, whereas no $O(1)$-approximation for EMD can be computed in $o(n^2)$ time \cite{buadoiu2005facility}. Essentially, this is due to the fact that MST cost is a more \emph{robust} problem than EMD. Indeed, in EMD increasing a single entry in the distance matrix can increase the EMD arbitrarily, whereas for MST this does not happen because of triangle inequality.

A valuable take-home message from this work is that allowing additive approximation makes EMD more robust. A natural question is whether we can find a $\varepsilon$-additive approximation to EMD in $\tilde O_\varepsilon(n)$ time, thus matching the above result on MST cost. 
The $\Omega(n^{1.2})$ lower bound on max-cardinality matching from~\cite{behnezhad2023sublinear} suggests that this should not be possible%
\footnote{The lower bound of~\cite{behnezhad2023sublinear} is proven in a slightly different model of adjacency list.}
Indeed, we can reduce max-cardinality matching to EMD by embedding the bipartite graph into a $(1,2)$ metric space. 
\section{Technical Overview}
\label{sec:tech overview}

Computing Earth Mover's Distance between two sets of $n$ points in a metric space can be achieved by solving Min-Weight Perfect Matching (MWPM) on the complete bipartite graph where edge-costs are given by the metric $d(\cdot, \cdot)$. Here we seek a suitable notion of approximation for MWPM that recovers \Cref{thm:main theorem intro EMD}.

\paragraph{Min-weight perfect matching with outliers.}
Consider the following problem: given a constant $\gamma > 0$, find a matching $M$ of size $(1-\gamma) n$ in a bipartite graph such that the cost of $M$ is at most the minimum cost of a perfect matching. 
A natural interpretation of this problem is to label a $\gamma$ fraction of vertices as outliers and leave them unmatched; so we dub this problem MWPM \emph{with outliers}.

Assuming $d(\cdot, \cdot) \in [0, 1]$, solving MWPM with a $\gamma$ fraction of outliers immediately yields a $\gamma$ additive approximation to EMD, proving \Cref{thm:main theorem intro EMD}. 

The main technical contribution of this work is the following theorem, which introduces an algorithm that solves MWPM with outliers in sublinear time. For the sake of this overview, the reader should instantiate \Cref{thm:main theorem mwm} with $\beta =1$ and think of $\gamma = (1-\alpha)$ as the fraction of allowed outliers.

\begin{restatable}{theorem}{MainThmMWM} 
\label{thm:main theorem mwm}
For each constants $0 \leq \alpha < \beta \leq 1$ there exists a constant $\varepsilon > 0$ and an algorithm running in time $O(n^{2 - \varepsilon})$ with the following guarantees.

The algorithm has adjacency-matrix access to an undirected, bipartite graph $G=(V_0 \cup V_1, E)$ and random access to the edge-cost function $c:E \rightarrow \bbR^+$. The algorithm returns $\hat{c}$ such that, whp,
\[
c(M^\alpha) \leq \hat c \leq c(M^\beta)
\]
where $M^\alpha$ is a minimum-weight matching of size $\alpha n$ and $M^\beta$ is a minimum-weight matching of size $\beta n$.

Moreover, the algorithm returns a matching oracle data structure that, given a vertex $u$ returns, in $n^{1+f(\varepsilon)}$ time, an edge $(u, v) \in \hat M$ or $\bot$ if $u \not\in V(\hat M)$, where $f(\varepsilon) \rightarrow 0$ when $\varepsilon \rightarrow 0$.
The matching $\hat M$ satisfies $\alpha n \leq |\hat M| \leq \beta n$ and $c(M^\alpha) \leq c(\hat M) \leq c(M^\beta)$.
\end{restatable}

Notice that the algorithm in \Cref{thm:main theorem mwm} does \textit{not} output the matching $\hat M$ explicitly. However, it returns a matching oracle data structure which implicitly stores $\hat M$.
The rest of this overview sketches the proof of \Cref{thm:main theorem mwm}.

\paragraph{Our algorithm, in a nutshell.}
A new set of powerful techniques was recently developed to approximate the size of a max-cardinality matching in sublinear time \cite{behnezhad2022time, behnezhad2023beating, behnezhad2023sublinear, BKS23a, bhattacharya2023dynamic}.
Our main contribution is a sublinear-time algorithm which leverages the techniques above to implement (a certain step of) the classic Gabow-Tarjan \cite{gabow1989faster} algorithm for MWPM. Since the techniques above return \emph{approximate} solutions, the obtained matching will be approximate as well, in the sense that we have to disregard a fraction of outliers when computing its cost to recover a meaningful guarantee.
Careful thought is required for relaxing the definitions of certain objects in the Gabow-Tarjan algorithm so as to accommodate their computation in sublinear time. The bulk of our analysis is devoted to proving that these relaxations combine well and lead to the guarantee in \Cref{thm:main theorem mwm}.

\paragraph{Roadmap.}
First, we will review (a certain step of) the Gabow-Tarjan algorithm that we will use as our template algorithm to be implemented in sublinear time. Then, we will review some recent sublinear algorithms for max-cardinality matching. Finally, we will sketch how to combine these tools to approximate the value of minimum-weight matching.

\subsection{A Template Algorithm}
The original Gabow-Tarjan algorithm operates on several scales and this makes it (slightly) more involved. We focus here on a simpler case where all our edge weights are integers in $[1, C]$, for $C=O(1)$. We will see in \Cref{sec:general costs and emd} that we can reduce to this case (incurring a small additive error). Here we describe our template algorithm, at a high level. 

\paragraph{A linear program for MPWM.}
First, recall the linear program for MWPM together with its dual. Here we consider a bipartite graph $G(V = V_0 \cup V_1, E)$ and cost function $c(\cdot, \cdot) \in [1, C]$.
We can interpret the following LP so that $x_{u, v} = 1$ iff $u$ and $v$ are matched, whereas primal constraints require every vertex to be matched.

\vspace{.4cm}

\begin{minipage}[t]{0.45\textwidth}
\textbf{Primal}
\begin{flalign*}
&\text{Minimize} &\sum_{u \in V_0, v \in V_0} x_{u, v} \cdot c(u, v) \\
&\text{subject to} &\sum_{v \in V_1} x_{u, v} \geq 1 \quad \forall u \in V_0 \\
& &\sum_{u \in V_0} x_{u, v} \geq 1 \quad \forall v \in V_1 \\
& & x_{u, v} \geq 0 \quad \forall u \in V_0 ,\; v \in V_1.
\end{flalign*}
\end{minipage}
\hfill\vline\hfill    
\begin{minipage}[t]{0.45\textwidth}
\textbf{Dual}
\begin{align*}
&\text{Maximize} &\sum_{u \in V} \varphi_u &\\
&\text{subject to} & \varphi_u + \varphi_v \leq c(u, v) \quad & \forall (u, v) \in E\\
& & \varphi_u \geq 0. \quad & \forall u \in V ,
\end{align*}
\end{minipage}

\vspace{.5cm}

\paragraph{A high-level description.}
Essentially, our template algorithm is a primal-dual algorithm which (implicitly) maintains a pair $(M, \varphi)$, where $M$ is a partial matching (so primal infeasible), and $\{\varphi(v)\}_{v \in V}$ is a vertex potential function, or an (approximately) feasible dual solution. Moreover, for each $e \in M$ the dual constraint corresponding to $e$ is tight. In other words, the pair $(M, \varphi)$ satisfies complementary slackness. The algorithm progressively grows the dual variables $\{\varphi(v)\}_{v \in V}$ and the size of $M$. When $M$ has size $\geq (1- \gamma )n$ then we are done. Indeed, throwing out $\gamma n$ vertices (as well as their associated primal constraints) we have that $(M, \varphi)$ is a (approximately) feasible primal-dual pair that satisfies complementary slackness, thus it is (approximately) optimal.

\paragraph{The primal-dual algorithm.}
We maintain an initially empty matching $M$. Inspired by the dual, we define a potential function $\varphi: V \rightarrow \bbZ$ and we enforce a relaxed version of the dual constraints: $\varphi(u) + \varphi(v) \leq c(u, v) + 1$ for each $(u, v) \in E$. Moreover, we maintain that $\varphi(u) + \varphi(v) = c(u, v)$ for each $(u, v) \in M$ (complementary slackness).
Let $T$ be the set of edges s.t.~the constraints above are tight. Orient the edges in $T$ so that all edges in $M \subseteq T$ are oriented from $V_0$ to $V_1$ and all edges in $T\setminus M$ are oriented from $V_1$ to $V_0$.
We denote the set of free (unmatched) vertices $F$ and let $F_0 = F \cap V_0$ $F_1 = F \cap V_1$.
We say that a path $P =(v_0 \rightarrow \cdots \rightarrow v_1)$ is an augmenting path if $v_0 \in F_0$, $v_1 \in F_1$ and $P$ alternates between edges in $T \setminus M$ and $M$. When we say that we augment $M$ wrt $P$ we mean that we set $M \leftarrow M \oplus P$.
We alternate between the following two steps:
\begin{enumerate}
    \item Find a a large set of node-disjoint augmenting paths $\{P_1 \dots P_\ell\}$. Augment $M$ wrt these paths. Decrement $\varphi(v) \texttt{ -= } 1$ for each $v \in \bigcup_i P_i \cap V_1$, to ensure the relaxed dual constraints are satisfied.
    \item Define $R$ as the set of vertices that are $T$-reachable\footnote{Recall that $T$ is oriented.} from $F_0$. Increment $\varphi(r_0) \texttt{ += } 1$ for each $r_0 \in R \cap V_0$, and decrement $\varphi(r_1) \texttt{ -= } 1$ for each $r_1 \in R \cap V_1$. This preserves the relaxed dual constraints and (eventually) adds some more edges to $T$.
\end{enumerate}
After $O_{\gamma, C}(1)$ iterations, we have $|F| \leq \gamma n$.

\paragraph{Analysis sketch.}
It is routine to verify that steps 1 and 2 preserve the relaxed dual constraints. At any point the pair $(M, \varphi)$ satisfies, $c(M) \leq \sum_{v \in V_0 \cup V_1} \varphi(v) \leq c(M') + n$ for any perfect matching $M'$. We can content ourselves with this additive approximation; indeed in \Cref{sec:general costs and emd} we will see how to charge it on the outliers.
To argue that we have few free vertices left after $O_\gamma(1)$ iterations, 
notice that at iteration $t$ we have $\varphi|_{F_0} \equiv t$ and $\varphi|_{F_1} \equiv 0$. Computing a certain function of potentials along $(M \oplus M')$-augmenting paths shows that $|F| \cdot t \leq O(n)$. Thus, $O_\gamma(1)$ iterations are sufficient to obtain $|F| \leq \gamma n$.
The arguments above are sufficient to show that our template algorithm finds an (almost) perfect matching with (almost) minimum weight. We will shove both \emph{almost} under the outlier carpet in \Cref{sec:general costs and emd}.

\subsection{Implementing the Template in Sublinear Time}

Our sublinear-time implementation of the template algorithm hinges on matching oracles.

\paragraph{Matching oracles.} Given a matching $M'$ we define a \emph{matching oracle} for $M'$ as a data structure that given $u \in V$ returns $v \in V$ if $(u, v) \in M'$ and $\bot$ otherwise. 
Note that given a matching oracle for $M'$, if we are promised that $|M'| = \Omega(n)$ then $O_\gamma(\log n)$ calls to such oracle are enough to estimate $|M'| \pm \gamma n$. 
We stress that all matching oracles that we use have sublinear query time.

\paragraph{Finding large matchings in sublinear time.}
An important ingredient in our algorithm is the $\largematch(G, A, \varepsilon, \delta)$ subroutine (\Cref{thm:large matching}), which is due to~\cite{bhattacharya2023dynamic}.  
Given $A \subseteq V$, $\largematch(G, A, \varepsilon, \delta)$ returns either $\bot$ or a matching oracle for some matching $M'$ in $G[A]$.
If there exists a matching in $G[A]$ of size $\delta n$, then $\largematch$ returns a matching oracle for some $M'$ in $G[A]$ with $|M'| = \Omega_\delta(n)$.
Else, if there are no matchings of size $\delta n$ in $G[A]$ $\largematch$ returns $\bot$. The parameter $\varepsilon$ controls the running time and essentially guarantees that $\largematch$ runs in $O(n^{2-\varepsilon})$ time while the matching oracle it outputs runs in $O(n^{1+\varepsilon})$.

We will use $\largematch$ to implement both step 1 and step 2 in the template algorithm. However, this requires us to relax our notions of maximal set of node-disjoint augmenting paths, as well as that of reachability. A major technical contribution of this work is to find the right relaxation of these notions so that:
\begin{enumerate}[1)]
    \item We can analyze a variant of the template algorithm working with these relaxed objects and still recover a a solution which is optimal if we neglect a $\gamma$ fraction of outliers.
    \item We can compute these relaxed objects in sublinear time using $\largematch$ as well as previously constructed matching oracles. 
\end{enumerate}
These relaxed notions are introduced in \Cref{sec:preliminaries}, point (1) is proven in \Cref{sec:template} and point (2) is proven in \Cref{sec:implementing in sublinear time}.

\paragraph{Implementing step 1 in sublinear time.}
 In \cite{bhattacharya2023dynamic} the authors implement McGregor's algorithm \cite{mcgregor2005finding} for streaming Max-Cardinality Matching (MCM) in a sublinear fashion using $\largematch$ (see \Cref{thm:augmenting paths} in this work). McGregor's algorithm finds a size-$\Omega(n)$ set of node-disjoint augmenting paths of fixed constant length, whenever there are at least $\Omega(n)$ of them. This notion is weaker than that of a maximal node-disjoint set of augmenting paths required in step 1 of our template algorithm in two regards: first, it only finds augmenting paths of fixed constant length; second, it finds only a constant fraction of such paths (as long as we have a linear number of them). 

In our template algorithm, the invariant $\varphi|_{F_1} \equiv 0$ is maintained (in step 2) because $R \cap F_1 = \emptyset$. In turn, $R \cap F_1 = \emptyset$ holds exactly because in step 1 we augment $M$ with a maximal node-disjoint set of augmenting paths.  
Since our sublinear implementation of step 1 misses some augmenting paths, the updates performed in step 2 will violate the invariant $\varphi(v)= 0$ for some $v \in F_1$. 

A careful implementation of step 2 (see next paragraph) guarantees that only missed augmenting paths that are short lead to a violation of $\varphi|_{F_1} \equiv 0$. Moreover, repeatedly running the sublinear implementation of McGregor's algorithm from \cite{bhattacharya2023dynamic}, we ensure that we miss at most $\gamma n$ short paths, for $\gamma$ arbitrary small. 
Thus, we can flag all vertices that belong to missed short augmenting paths as outliers since we have only a small fraction of them.

\paragraph{Implementing step 2 in sublinear time.}
We implement an approximate version of the reachability query in step 2 as follows. We initialize the set of reachable vertices $R$ as $R \leftarrow F_0$. Then, for a constant number of iterations: we compute a large matching $M' \subseteq T \setminus M$ between the vertices of $R \cap V_0$ and $V_1 \setminus R$; then we add to $R$ all matched vertices in $\bigcup M'$ as well as their $M$-mates, namely $\mate_M(u)$ for each $u \in \bigcup M'$.
Notice that if a $\Omega(n)$-size matching $\subseteq T \setminus M$ between $R \cap V_0$ and $V_1 \setminus R$ exists, then we find a matching $\subseteq T \setminus M$ between $R \cap V_0$ and $V_1 \setminus R$ of size at least $\Omega(n)$. This ensures that: (i) after a constant number of iterations $\largematch$ returns $\bot$; (ii) when $\largematch$ returns $\bot$ there exists a vertex cover $\cC$ of $((R \cap V_0) \times (V_1 \setminus R)) \cap T \setminus M$ of size $\gamma n$.
Only constraints corresponding to edges incident to $\cC$ might be violated during step 2. Furthermore, $|\cC| = \gamma n$ is small and so we can just label vertices in $\cC$ as outliers.

As we pointed out in the previous paragraph, the invariant $\varphi|_{F_1} \equiv 0$ might be violated in step 2 if $R \cap F_1 \neq \emptyset$. 
We already showed that whenever we the missed augmenting path causing the violation of $\varphi|_{F_1}  \equiv 0$ is short we can charge this violation on a small set of outliers.
To make sure that no long augmenting path leads to a violation of $\varphi|_{F_1} \equiv 0$ we set our parameters so that the depth of the reachability tree built in step 2 is smaller than the length of ``long'' paths. Thus, any long path escapes $R$ and cannot cause a violation.

\paragraph{Everything is an oracle.}
The implementation of both step 1 and step 2 operates on the graph $T$ of tight constraints. To evaluate $(u, v) \in T$, we need to compute $\varphi(u)$ and $\varphi(v)$. In turn, the potential values depend on previous iterations of the algorithm. None of these iterations outputs an explicit description of the objects described in the template (potentials, matchings, augmenting paths or sets of reachable vertices). Indeed, these objects are output as oracle data structures, which internals call (eventually multiple) matching oracles output by $\largematch$.
We prove that essentially all these oracles have query time $O(n^{1+\varepsilon})$ for some small $\varepsilon > 0$. A careful analysis is required to show that we can build the oracles at iteration $i + 1$ using the oracles at iteration $i$ without blowing up their complexity.

\paragraph{Paper organization.}
In \Cref{sec:preliminaries} we define some fundamental objects that we will use throughout the paper. In \Cref{sec:template} we present a template algorithm to be implemented in sublinear time, and prove its correctness. In \Cref{sec:implementing in sublinear time} we implement the template algorithm in sublinear time. In \Cref{sec:general costs and emd} we put everything together and prove the main theorems stated in the introduction.
\section{Preliminaries}
\label{sec:preliminaries}

We use the notation $[a, b] := \{a \dots b-1\}$, $[b] = [0, b]$, and $(a \pm b) := [a - b, a + b]$ meaning that $c \cdot (a\pm b) = (ac \pm bc)$. We denote our undirected bipartite graph with $G=(V, E)$, and the bipartition is given by $V = V_0 \cup V_1$. Our original graph is complete and for each $(u, v) \in V_0 \times V_1$ we denote with $c(u, v)$ the cost of the edge $(u, v)$.  We stress that none of our algorithms require $c(\cdot, \cdot)$ to be a metric. Given a matching $M$ we denote its combined cost with $c(M)$. For each $u \in V$ we say that $u = \mate_M(v)$ iff $(u, v) \in M$. When the matching $M$ is clear from the context we denote with $F$ the set of unmatched (or \emph{free}) vertices, and set $F_i := F \cap V_i$ for $i=0,1$. 

When we say that an algorithm runs in time $t$ we mean that both its computational complexity and the number of queries to the cost matrix $c(\cdot, \cdot)$ are bounded by $t$. The computational complexity of our algorithms is always (asymptotically) equivalent to their query complexity, so we only analyse the latter. All our guarantees in this work hold with high probability.

\begin{definition}[Augmenting paths]
Given a matching $M$ over $G=(V, E)$ we say that $P = (v_0, v_1 \dots v_{2\ell+1})$ is an augmenting path w.r.t. $M$ if $(v_{2i}, v_{2i+1}) \in E \setminus M$ for each $i = 0 \dots \ell$ and $(v_{2j+1}, v_{2j+2}) \in M$ for each $j = 0 \dots \ell - 1$.  
When we say that we augment $M$ w.r.t. $P$ we mean that we set $M \leftarrow M \oplus P$, where $\oplus$ is the exclusive or.
\end{definition}

We use the same notion of $1$-feasible potential as in  \cite{gabow1989faster}.
\begin{definition}[$1$-feasibility conditions]
Given a potential $\varphi: V \longrightarrow \bbZ$ we say that it satisfies $1$-feasibility conditions with respect to a matching $M$ if the following hold.
\begin{enumerate}[(i)]
    \item For each $u \in V_0, v \in V_1$ $\varphi(u) + \varphi(v) \leq c(u, v) + 1$.
    \item For each $(u, v) \in M$, $\varphi(u) + \varphi(v) = c(u, v)$.
\end{enumerate}    
\end{definition}

\begin{definition}[Eligibility Graph]
We say that an edge $(u, v)$ is eligible w.r.t. $M$ if: $(u, v) \not\in M$ and $\varphi(u) + \varphi(v) = c(u, v) + 1$ or; $(u, v) \in M$ and $\varphi(u) + \varphi(v) = c(u, v)$.
We define the eligibility graph as the directed graph $G_\cE = (V, E_\cE)$ that orients the eligible edges so that, for each eligible $(u, v) \in V_0 \times V_1$, we have $(u, v) \in E_\cE$ if  $(u, v) \not \in M$ and $(v, u) \in E_\cE$ if $(u, v) \in M$.
\end{definition}

Notice that, whenever a potential is $1$-feasible w.r.t. $M$, then all edges in $M$ are eligible.

\begin{definition}[Forward Graph]
We define the forward graph $G_F = (V, E_F)$ as the subgraph of the eligibility graph containing only edges from $V_0$ to $V_1$. That is, we remove all edges $(v, u)$ such that $(u, v) \in M$.
\end{definition}

Now, we introduce two quite technical definitions, which provide us with approximate versions of the notion of ``maximal set of node-disjoint augmenting paths'' and ``maximal forest''.

\begin{definition}[$(k, \xi)$-Quasi-Maximal Set of Node-Disjoint Augmenting Paths]
\label{def:qmsndap}
Given a graph $G=(V, E)$ and a matching $M \subseteq E$ we say that a set $\cP$ of augmenting paths of length at most $k$ is a $(k, \xi)$-\qmsndap if for any $\cQ$ such that $\cQ \cup \cP$ is a set of node-disjoint augmenting paths of length $\leq k$ we have $|\cQ| \leq \xi n$.
\end{definition}

Intuitively, $\cP$ is a $(k, \xi)$-\qmsndap if we can add only a few more node-disjoint augmenting paths of length $\leq k$ to $\cP$ before it becomes a maximal.
Next we introduce an approximate notion of ``maximal forest'' $\cF$ in the eligibility graph $G_\cE$ rooted in $F_0$. $\cF$ is obtained starting from the vertices in $F_0$ and adding edges (in a way that we will specify later) so as to preserve the $\cF$ has $|F_0|$ connected components and has no cycles.
This construction will ensure that the connected component of our forest have small diameter and small size.
We maintain that whenever $v \in V_1$ is added to $\cF$, then $\mate_M(v)$ is also added to $\cF$.
$\cF$ is approximately maximal in the sense that the cut $(\cF, V\setminus \cF)$ in $G_\cE$ admits a small vertex cover.

\begin{definition}[$(k, \delta)$-Quasi-Maximal Forest]
\label{def:qmf}
Given the eligibility graph $G_\cE=(V, E_\cE)$ w.r.t. the matching $M$, and the set of vertices $F_0 \subseteq V_0$ we say that $\cF$ is a $(k, \delta)$-\qmf  ~rooted in $F_0$ if:
\begin{enumerate}
    \item $F_0 \subseteq \cF$ \label{enum:forest def free included}
    \item For each $v \in V_1 \cap \cF$ we have $\mate_M(v) \in \cF$ \label{enum:forest def matched edges included}
    \item For each $u \in \cF$ there exists $v \in F_0$ at hop distance from $u$ at most $k$. \label{enum:forest def forest depth}
    \item Every connected component of $\cF$ has size at most $2^k$. \label{enum:forset def cc size}
    \item The edge set $E_\cE \cap (\cF \times V \setminus \cF)$ has a vertex cover of size at most $\delta n$. \label{enum:forest def vertex cover}
\end{enumerate}    
\end{definition}

Now, we introduce a few results from past work on sublinear-time maximum caridnality matching. The following theorem, which is the main technical contribution of \cite{bhattacharya2023dynamic}, states that we can compute a $\varepsilon n$-additive approximation of the size of a maximum-cardinality matching in strongly sublinear time.

\begin{theorem}[Theorem 1.3, \cite{bhattacharya2023dynamic}]
\label{thm:mcm in sublinear}
There is a randomized algorithm that, given the adjacency matrix of a graph $G$, in time $n^{2-\Omega_{\varepsilon}(1)}$ computes with high probability a $(1,\varepsilon n)$-approximation $\tilde{\mu}$ of $\mu(G)$. 
After that, given a vertex $v$, the algorithm returns in $n^{1+f(\varepsilon)}$ time an edge $(v,v')\in M$ or $\bot$ if $v\notin V(M)$ where $M$ is a fixed $(1,\varepsilon n)$-approximate matching, where $f$ is an increasing function such that $f(\varepsilon) \rightarrow 0$ when $\varepsilon\rightarrow 0$. 
\end{theorem}

The algorithm in \Cref{thm:mcm in sublinear} does not exactly output a matching, but rather a \emph{matching oracle}. Namely, it outputs a data structure that stores a matching $M$ implicitly. We formalize the notion of matching oracle below. 

\begin{definition}[Matching Oracle]
Given a matching $M$, we define the matching oracle $\match_M(\cdot)$ as a data structure such that $\match_M(u) = v$ if $(u, v) \in M$ and $\match_M(u) = \bot$ otherwise. Throughout the paper we denote with $t_M$ the time complexity of $\match_M(\cdot)$. 
\end{definition}

Similarly to matching oracles, we make use of membership oracles $\mem_A(\cdot)$ and potential oracles $\eval_\varphi(\cdot)$ where $A \subseteq V$ and $\varphi$ is a potential function defined on $V$. As expected, $\mem_A(u)$ returns $\one_{u \in A}$ and $\eval_\varphi(u)$ returns $\varphi(u)$. We denote their running time with $t_A$ and $t_\varphi$ respectively.
Now, we recall two theorems from \cite{bhattacharya2023dynamic} that constitutes fundamental ingredients of our sublinear-time algorithm for minimum-weight matching. 

\Cref{thm:large matching} roughly says that, in sublinear time, we can find a matching oracle for a size-$\Omega(n)$ matching, whenever a size-$\Omega(n)$ matching exists. 

\begin{theorem}[Essentially Theorem 4.1, \cite{bhattacharya2023dynamic}]
\label{thm:large matching}
Let $G=(V,E)$ be a graph, $A\subseteq V$ be a vertex set. Suppose that we have access to adjacency matrix of $G$ and an $A$-membership oracle $\mem_{A}$ with $t_{A}$ query time. We are given as input a sufficiently small $\varepsilon>0$ and $\sizeIn>0$.

There exists an algorithm $\largematch(G, A, \varepsilon, \sizeIn)$ that preprocesses $G$ in $\Otil_{\sizeIn} ((t_{A}+n) \cdot n^{1-\varepsilon})$  time and either return $\bot$ or construct a matching oracle $\match_{M}(\cdot)$ for a matching $M\subset G[A]$ of size at least $\sizeOut n$ where $\sizeOut=\frac{1}{2000}\sizeIn^{5}$ that has $\Otil_{\sizeIn} ((t_{A}+n)n^{4\varepsilon})$ 
worst-case query time. If $\mu(G[A])\ge\sizeIn n$, then $\bot$ is not returned. The guarantee holds with high probability.
\end{theorem}

\Cref{thm:augmenting paths} roughly says that, in sublinear time, we can increase the size of our current matching (oracle) by $\Omega(n)$, whenever there are $\Omega(n)$ short augmenting paths. 

\begin{theorem} [Essentially Theorem 5.2, \cite{bhattacharya2023dynamic}]
\label{thm:augmenting paths} 
Fix two constants $k, \gamma > 0$.
For any sufficiently small $\varepsilon_{\inp} > 0$, there exists $\varepsilon_{\out} = \Theta_{k, \gamma}(\varepsilon_{\inp})$ such that the following holds. 
There exists an algorithm $\augment(G, M^{\inp}, k, \gamma, \varepsilon_{\inp})$  that
makes $O_{k, \gamma}(1)$ calls to $\largematch$ which take $\tilde{O}_{k, \gamma}\left(n^{2-\varepsilon_{\inp}}\right)$ time in total. Further, either it returns an oracle $\match_{M^{\out}}(\cdot)$ with query time $\tilde{O}_{k,\gamma}(n^{1+\varepsilon_{\out}})$, for some matching $M^{\out}$ in $G$ of size $|M^{\out}| \geq |M^{\inp}| + \Theta_{k, \gamma}(1) \cdot n$ (we say that it ``succeeds'' in this case), or it returns {\sc Failure}.  Finally, if the matching $M^{\inp}$ admits a collection of $\gamma \cdot n$ many node-disjoint length $(2k+1)$-augmenting paths in $G$, then the algorithm succeeds whp.
\end{theorem}

\Cref{thm:augmenting paths} differs from Theorem 5.2 in \cite{bhattacharya2023dynamic} in that it specifies that the only way $\augment$ accesses the graph is through $\largematch$. We will use this property crucially to prove \Cref{lem:augmenting paths with potential}.
\section{A Template Algorithm}
\label{sec:template}

In this section we study min-weight matching with integral small costs $c(\cdot, \cdot) \in [1, C]$, where $C$ is constant. We will see how to lift this restriction in \Cref{sec:general costs and emd}. \Cref{alg:template algorithm} gives a template algorithm realising \Cref{thm:main theorem mwm} that assumes we can implement certain subroutines; in \Cref{sec:implementing in sublinear time} we will see how to implement these subroutines in sublinear time.

\paragraph{Comparison with Gabow-Tarjan.}
Intuitively, our template algorithm implements the Gabow-Tarjan algorithm \cite{gabow1989faster} for a fixed scale in an approximate fashion. Indeed, instead of finding a maximal-set of node-disjoint augmenting paths we find a $(k, \xi)$-\qmsndap and instead of growing a forest in the eligibility graph we grow a $(k, \delta)$-\qmf. See \Cref{fig:QMSNDAP and QMF} for a representation of step 1 and step 2. 

\begin{algorithm}
\textbf{Input:} A bipartite graph $G=(V_0 \cup V_1, E)$ and a cost function $c:E \rightarrow [1, C]$.

Set $T = C / \gamma^3$, $\xi = \frac{\gamma}{T k 2^{k}}$, $\delta = \frac{\gamma}{T}$ and $k = 6000 (2T+1)^{10} / \delta^5 $.

Initialize $M \leftarrow \emptyset$ and $\varphi(v) \leftarrow 0$ for each $v\in V$.

Let $F_0$ denote the set of $M$-unmatched vertices in $V_0$.

For each $e \in E$ update $c(e) \leftarrow c(e) / \gamma$ (this is implemented lazily).

Execute the following two steps for $T$ iterations:
\begin{itemize}
    \item \textit{Step 1.} Find a $(k, \xi)$-\qmsndap $\cP$ in the eligibility graph $G_\cE$. Augment $M$ w.r.t. paths in $\cP$. Set $\varphi(v) \leftarrow \varphi(v) - 1$ for each $v \in V_1 \cap \bigcup_{P\in \cP} P$.
    \item \textit{Step 2.} Find a $(k, \delta)$-\qmf $\cF$ rooted in $F_0$ in the eligibility graph $G_\cE$. Set $\varphi(u) \leftarrow \varphi(u) + 1$ for each $u \in V_0 \cap \cF$ and $\varphi(v) \leftarrow \varphi(v) - 1$ for each $v \in V_1 \cap \cF$.
\end{itemize}

Sample a set $S$ of $O_{\gamma, C}(\log n)$ edges in $M$ with replacement.

Discard the $3\gamma |S|$ edges with highest costs and let $\Sigma$ be the sum of costs of remaining edges. 

\textbf{Output:} $\hat c = \frac{n}{|S|} \Sigma$.

\caption{Template Algorithm. \label{alg:template algorithm}}
\end{algorithm}

\begin{figure}[!ht]
\centering
\includegraphics[page=2, scale=.35]{figures/EMD-figures.pdf}
\hspace{.7cm}
\includegraphics[page=1, scale=.35]{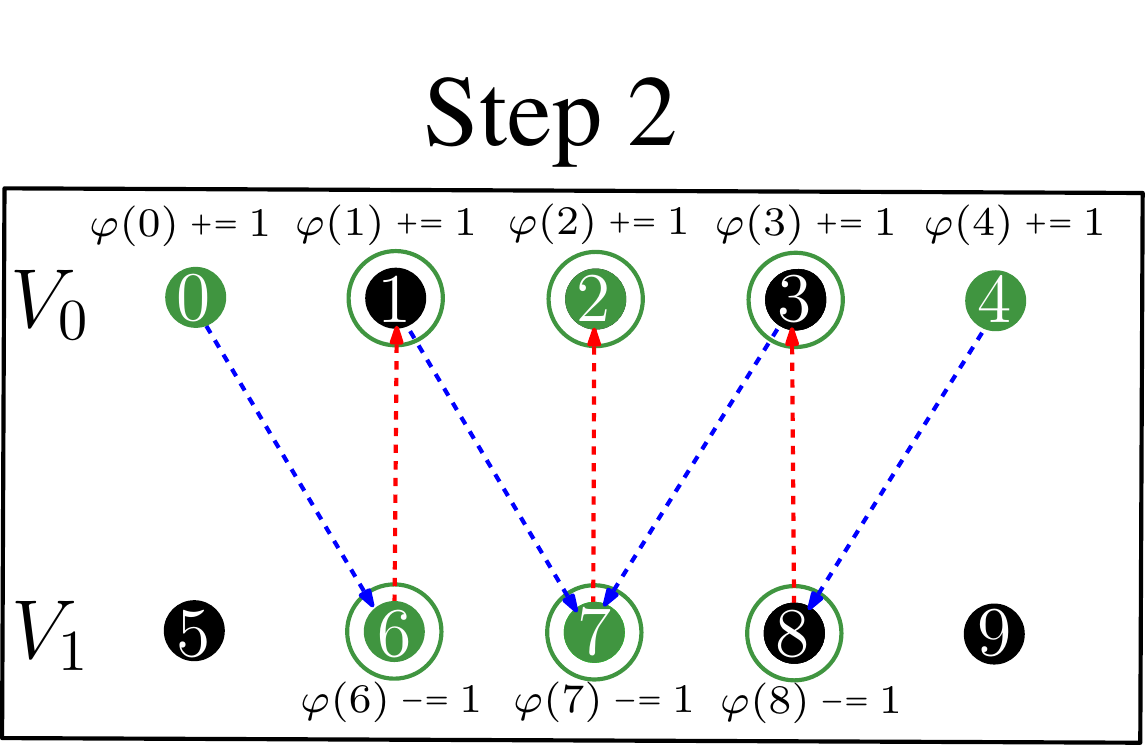}
\caption{
We color the edges of $M$ in red and the edges of $T\setminus M$ in blue.
On the left we have an example of step 1. Solid edges represent paths in the QMSNDAP $\cP$ that we augment $M$ along in step 1.
On the right we have an example of step 2. 
All vertices colored or circled in green belong to the QMF $\cF$. 
Circles help us visualize the implementation of step 2, described in \Cref{sec:implementing in sublinear time}. In \Cref{alg:implementation of step2} $\cF$ is built sequentially, where each iteration (lines 1-5) adds some edges to $\cF$. At first, only the non-circled green vertices belong to $\cF$. The first step adds the green-circled black edges, and the second step adds the green-circled green edges.   
}
\label{fig:QMSNDAP and QMF}
\end{figure}

\paragraph{Analysis.}
Here we analyse \Cref{alg:template algorithm} and show that it satisfies the following theorem.
\begin{theorem}
\label{thm:weaker main thm mwm}
Fix a constant $\gamma > 0$.
Suppose that we have adjacency-matrix access to the bipartite graph $G=(V_0 \cup V_1, E)$ and random access to the cost function $c:E \rightarrow [1, C]$, with $C = O(1)$. Then, with high probability, \Cref{alg:template algorithm} returns $\hat c$ such that
\[
c(M^{1-\gamma}) \leq \hat c \leq c(M^\opt)
\]
where $M^{1-\gamma}$ is a min-weight matching of size $(1-\gamma) n$  and $M^\opt$ is a min-weight matching of size $n$.
\end{theorem}

To prove \Cref{thm:weaker main thm mwm}, we need a series of technical lemmas.
\paragraph{Proof Roadmap.}
The proof of \Cref{thm:weaker main thm mwm} goes as follows.
We prove that, after $T$ iterations, all free vertices in $F_0$ have potential $T$. On the other hand, the majority of free vertices in $F_1$ have potential $0$. We call \emph{spurious} the free vertices in $F_1$ with non-zero potential and we show there are only few of them. Then, (roughly) we look at the final matching $M$ generated by \Cref{alg:template algorithm} and a perfect matching $M'$ and consider the graph $G'$ having $M \oplus M'$ as its set of edges. $G'$ can be partitioned into cycles and augmenting paths. Each augmenting path starts in a free vertex in $F_0$ and ends in a free vertex in $F_1$. If the $1$-feasibility conditions are satisfied by all edges, then computing a certain function of potentials along an augmenting path and combining the results for all augmenting paths yields an upper bound on the total number of free vertices. Unfortunately, not all edges satisfy the $1$-feasibility constraints. We fix this by finding a small vertex cover of the $1$-unfeasible edges. We say that such cover a suitable set of \emph{broken} vertices. Ignoring spurious and broken vertices is sufficient to make our argument work.

\begin{lemma}
\label{lem:free v0 vertices potential}
After $t \in [T+1]$ iterations we have $\varphi(u) = t$ for each $u \in F_0$.
\end{lemma}
\begin{proof}
After $t=0$ iterations, we have $\varphi(u) = 0$ for each $u \in V$.
First, we notice that the set of unmatched (or free) vertices $F$ only shrinks over time, and so does $F_0$. Moreover, at each iteration we increase the potential of free vertices in $F_0$ by $1$.
\end{proof}

Define the set $S$ of $v \in F_1$ such that $\varphi(v) \neq 0$ as the set of \emph{spurious} vertices.

\begin{lemma} \label{lem:few spurious}
After $T$ iterations we have at most  $\gamma n$ \emph{spurious} vertices.
\end{lemma}
\begin{proof}
We prove that at each iteration we increase the number of spurious vertices by at most $\gamma n/T$. A vertex cannot become spurious in Step 1. Indeed, in Step 1 we only decrease the potential of matched vertices. If a vertex $v \in F_1$ becomes spurious in Step 2, it means that there exists an augmenting path $P$ from some $u \in F_0$ to $v$ contained in a connected component of $\cF$.
Let $\cQ$ be such that $\cQ \cup \cP$ is a maximal set of node-disjoint augmenting paths of length $\leq k$. By \Cref{def:qmsndap} we have $|\cQ| \leq \xi n$. Define the set of \emph{forgotten} vertices as $\bigcup_{Q \in \cQ} Q$.
Thanks to \cref{enum:forest def forest depth} in \Cref{def:qmf}, the path from $u$ to $w$ has length $\leq k$, thus $P$ has length at most $k$. Recall that $P$ is an augmenting path w.r.t. the graph obtained augmenting $M$ along $\cP$ at the end of Step 1. Therefore, $P$ intersects a path in $\cQ \cup \cP$. 

We now argue that that $P$ cannot intersect $P' \in \cP$. Suppose by contradiction that it does. 
Let $P = (P_0 \dots P_\ell)$ and $P' = (P'_0 \dots P'_{\ell'})$. Let $P_s$ be the first (w.r.t. the order induced by $P$) node where $P$ and $P'$ intersect. We first rule out the case that $s$ is even: for $s = 0$,  $P_0 = u \in F_0$ implies that $u$ did not belong to an augmenting path $P'$ in Step 1. Moreover, for $s = 2i >0$ if $P_{2i} = P'_j$ then $P_{2i-1} = \mate_M(P_{2i}) \in \{P'_{j-1}, P'_{j+1}\}$, where $M$ is the matching obtained at the end of Step 1.
Now suppose that $s$ is odd, and hence $P_s \in V_1 \cap P'$. Then $\varphi(P_s)$ is decreased by $1$ at the end of Step 1, hence no edge outside of $M$ incident to $P_s$ is eligible in Step 2.

Thus, $P$ must intersect a path in $\cQ$. On the other hand $\bigcup_{Q \in \cQ} Q$ contains at most $k\xi n$ vertices, so at most $k\xi n$ connected component of $\cF$ contain a forgotten edge. Moreover, by \cref{enum:forset def cc size} of \Cref{def:qmf} every connected component of $\cF$ has size at most $2^k$, thus at most $k\xi n 2^k = \gamma / T$ vertices become spurious.
\end{proof}

We say that $B\subseteq V$ is a suitable set of \emph{broken} vertices if all $(u, v) \in (V_0 \setminus B) \times (V_1 \setminus B)$ are $1$-feasible.

\begin{lemma} \label{lem:few broken vertices}
After $T$ iterations, there exists a suitable set of broken vertices of size at most $\gamma n$.
\end{lemma}
\begin{proof}
First, we prove that every edge $(u, v) \in V_0 \times V_1$, which is $1$-feasible at the beginning of Step 1, is also $1$-feasible at the end of Step 1. Suppose that $(u, v)$ becomes $1$-unfeasible in Step 1. Let $M$ and $M'$ be the matching at the beginning and at the end of Step 1 respectively. Potentials only decrease in Step 1, so in  order for $(u, v)$ to become $1$-unfeasible w.r.t. $M'$ we must have $(u, v) \in M'$. Moreover, we decrease the potential of $v$ only if $(u, v) \in P$, for some augmenting path $P$. Thus, at the beginning of Step 1 we had $\varphi(u) + \varphi(v) = c(u, v) + 1$, which implies $\varphi(u) + \varphi(v) = c(u, v)$ at the end of Step 1, thus $(u, v)$ is $1$-feasible w.r.t. $M'$, contradiction. 

Now, we grow a set of suitable broken vertices $B$. We initialize $B = \emptyset$ and show that each iteration Step 2 increases the size of $B$ by at most $\gamma n / T$.
If $(u, v)\in V_0 \times V_1$ is $1$-feasible at the beginning of Step 2 and becomes $1$-unfeasible in Step 2, then we must have $u \in \cF$ and $v \not\in \cF$.
Indeed, by \cref{enum:forest def matched edges included} in \Cref{def:qmf} if $(u, v) \in M'$ then either both $u$ and $v$ belong to $\cF$ or neither of them does. 
This ensures that the sum of their potentials is unchanged. Else, if $(u, v) \not\in M'$ then in order for it to violate $1$-feasibility we must increase $\varphi(u)$ by one and not decrease $\varphi(v)$, and this happens only if $u \in \cF$ and $v \not\in \cF$.
\Cref{enum:forest def vertex cover} in \Cref{def:qmf} ensures that there exists a vertex cover $U \subseteq V$ for the set of new $1$-unfeasible edges with $|U| \leq \delta n = \gamma n / T$. We update $B\leftarrow B \cup U$. Thus, after $T$ iterations we have $|B| \leq \gamma n$.
\end{proof}

\begin{lemma} \label{lem:few free vertices}
After $T$ iterations of template algorithm we have have $|F_0| = |F_1|  \leq 4 \gamma n$.
\end{lemma}
\begin{proof}
Denote with $M$ the final matching obtained by~\Cref{alg:template algorithm}. Let $B$ be a suitable set of broken vertices with $|B| \leq \gamma n$, as in \Cref{lem:few broken vertices}. Partition $B = B_M \cup B_F$, where $B_F:= B \cap F$ is the set of unmatched vertices in $B$ and $B_M$ is the set of matched vertices in $B$.
Consider the set $B'_M$ of vertices currently matched to vertices in $B_M$, $B'_M = \{\mate_{M}(b) \cond b \in B_M\}$. 
We have $|(B_M\cup B'_M) \cap V_0| = |(B_M\cup B'_M) \cap V_1| \leq \gamma n$. 
Let $S$ be the set of spurious vertices and recall that $|S| \leq \gamma n$ by \Cref{lem:few spurious}. 
Let $S' \subseteq F_0 \setminus B_F$ such that $|S' \cup (V_0 \cap B_F)| = |S \cup (V_1 \cap B_F)|$. This implies that $|S' \cup (V_0 \cap B_F)| \leq |S| + |B| \leq 2\gamma n$.
Define $A_0 := V_0 \setminus (B \cup B'_M \cup S')$ and $A_1 := V_1 \setminus (B \cup B'_M \cup S)$ and notice that they have the same size. Define $A = A_0 \cup A_1$. Let $M'$ be a perfect matching over $A$.

The graph $G_A = (A, M \oplus M')$ contains exactly $\ell := |F_0 \cap A_0| = |F_1 \cap A_1|$ node-disjoint paths $P_1 \dots P_\ell$ where $P_i$ starts in $f^{(i)}_0 \in F_0 \cap A_0$ and ends in $f_1^{(i)} \in F_1 \cap A_1$.
We define the \emph{value} of a path $P$ as
\begin{equation*}
\cV(P) = \sum_{(u, v) \in M' \cap P} \ld(c(u,v) + 1\rd) - \sum_{(u, v) \in M \cap P} c(u,v).
\end{equation*}
By $1$-feasibility of $\varphi$ we have
\[
\cV(P_i) \geq  \sum_{(u, v) \in M' \cap P} \ld(\varphi(u) + \varphi(v)\rd) - \sum_{(u, v) \in M \cap P} (\varphi(u) + \varphi(v)) = \varphi\ld(f^{(i)}_0\rd) - \varphi\ld(f_1^{(i)}\rd) = T
\]
where the last equality holds by definition of (non-)spurious vertices and \Cref{lem:free v0 vertices potential}.
Then, we have $C n \geq n + c(M') \geq \sum_{i}^\ell \cV(P_i) \geq \ell T$. Thus, $\ell \leq Cn/T = \gamma n$ and 
\[
|F_1| = |F_0| \leq |F_0 \cap A| + |(B_M\cup B'_M) \cap V_0| + |S' \cup (V_0 \cap B_F)| \leq 4\gamma n.
\]
\end{proof}

Let $\varphi$ be the potential at the end of the execution of \Cref{alg:template algorithm}.
Denote with $M^\alg$ the final matching obtained by \Cref{alg:template algorithm} and with $M^\opt$ a min-weight perfect matching.
Given a matching $M$, we denote with $M_{[\alpha]}$ the matching obtained from $M$ by removing the $\alpha n$ edges with highest cost. 

\begin{lemma}
\label{lem:cheaper matchings}
We have $c(M^\alg_{[2\gamma]}) \leq c(M^\opt)$.
\end{lemma}
\begin{proof}
Let $M_{\setminus B}^\alg$ be the matching obtained from $M^\alg$ by removing all edges incident to vertices in $B$. Since $|B| \leq \gamma n$ we have $c(M^\alg_{[\gamma]}) \leq c(M_{\setminus B}^\alg)$. Notice that all edges in $M_{\setminus B}^\alg$ are $1$-feasible.
For each $(u, v) \in M_{\setminus B}^\alg$ we have $c(u, v) = \varphi(u) + \varphi(v)$ and for each $(u, v) \in M^\opt$ we have $\varphi(u) + \varphi(v) \leq c(u, v) + 1$. Thus,
\[ 
c(M^\alg_{[\gamma]}) \leq c(M_{\setminus B}^\alg) \leq \sum_{u \in V} \varphi(u) = \sum_{(u, v) \in M^\opt} \varphi(u) + \varphi(v) \leq  \sum_{(u, v) \in M^\opt} c(u, v) + 1 = n + c(M^\opt). 
\]
Now, it is sufficient to notice that, since all edges have costs in $[1/\gamma, C/\gamma - 1]$, removing any $\gamma n$ edges from $M^\alg_{[\gamma]}$ decreases its cost by $n$. Thus, $c(M^\alg_{[2\gamma]}) \leq c(M^\opt)$.
\end{proof}

Now, we are ready to prove \Cref{thm:weaker main thm mwm}.

\begin{proof}[Proof of \Cref{thm:weaker main thm mwm}]
Thanks to \Cref{lem:cheaper matchings}, we know that $c(M^\alg_{[2\gamma]}) \leq c(M^\opt)$. Moreover, by \Cref{lem:few free vertices} we have $|M^\alg| = n - |F_0| \geq (1-4\gamma)n$, thus defining $M^{1 - 8\gamma}$ as the min-weight matching of size $(1-8\gamma)n$, we have $c(M^{1-8\gamma}) \leq c(M^\alg_{[4\gamma]})$.
We are left to prove that the estimate $\hat c = \frac{n}{|S|}\Sigma$ returned by \Cref{alg:template algorithm} satisfies $c(M^\alg_{[4\gamma]}) \leq \hat c \leq c(M^\alg_{[2\gamma]})$. 
Let $S$ and $\Sigma$ be defined as in \Cref{alg:template algorithm} and let $w$ be maximum such that $3\gamma|S|$ edges in $S$ have cost $\geq w$. 
If $\alpha_{w} \cdot n$ is the number of edges in $M^\alg$ that cost $\geq w$, then using standard Chernoff Bounds arguments we have that, whp, $|\alpha_{w} - 3 \gamma| \leq \gamma^2 / C$. From now on we condition on this event. Notice that $\frac{(1-\alpha_w) n}{(1-3\gamma) |S|} \Sigma$ is an unbiased estimator of $c(M^\alg_{[\alpha_w]})$. 
Moreover, since all costs are in $[1/\gamma, C/\gamma]$, then $O_{\gamma, C} (\log n)$ samples are sufficient to have $\frac{(1-\alpha_w) n}{(1-3\gamma) |S|} \Sigma$ concentrated, up to a factor $(1\pm \frac{\gamma^2}{C})$, around $c(M^\alg_{[\alpha_w]})$. 
Hence, assuming that $\gamma$ is sufficiently small, we have
\[
\frac{n}{|S|} \Sigma \in (1\pm 3\gamma^2/C) \cdot c(M^\alg_{[\alpha_w]}) \subseteq c(M^\alg_{[\alpha_w]}) \pm 3\gamma n
\]
where the last containment relation holds because all costs are $\leq C/\gamma$ and so $c(M^\alg_{[\alpha_w]}) \leq Cn/\gamma$.
Since all costs are $\geq 1/\gamma$ we have $c(M^\alg_{[\alpha_w + 3\gamma^2]}) \leq c(M^\alg_{[\alpha_w]}) - 3\gamma n$ and 
$c(M^\alg_{[\alpha_w - 3\gamma^2]}) \geq c(M^\alg_{[\alpha_w]}) + 3\gamma n$.
Thus, picking $\gamma$ small enough to have $\alpha_w \pm 3\gamma^2 \subseteq [2\gamma, 4\gamma]$ we have
\[
 c(M^\alg_{[4 \gamma]}) \leq \frac{n}{|S|} \Sigma \leq c(M^\alg_{[2 \gamma]}).
\]
Therefore, we have $c(M^{1-8\gamma}) \leq \hat c \leq c(M^\opt)$ and rescaling $\gamma$ gives exactly the desired result.
\end{proof}

\begin{observation}
\label{obs:how to build matching oracle}
As in the proof of \Cref{thm:weaker main thm mwm}, define $w$ as the maximum value such that there are at least $3 \gamma |S|$ edges with cost $\geq w$ in $S$ and define $\alpha_w$ such that exactly $\alpha_w \cdot n$ edges in $M$ have cost $\geq w$. 
We have, whp, $|\alpha_w -3\gamma| \leq \gamma^2 / C$, thus for $\gamma$ small enough $c(M^\alg_{[\alpha_w]}) \leq c(M^\alg_{[2\gamma]}) \leq c(M^\opt)$ and (up to rescaling $\gamma$) $|M^\alg_{[\alpha_w]}| \geq (1-\gamma) n$.
Moreover, given an edge $e \in M^\alg$ we can decide whether $e \in M^\alg_{[\alpha_w]}$ simply by checking $c(e) \leq w$.    
\end{observation}

\section{Implementing the Template in Sublinear Time} \label{sec:implementing in sublinear time}

In this section we explain how to implement \emph{Step 1} and \emph{Step 2} from the template algorithm in sublinear time.

\subsection{From Potential Oracles to Membership Oracles}

Throughout this section, we would like to apply \Cref{thm:large matching} and \Cref{thm:augmenting paths} on 
the eligibility graph $G_\cE = (V, E_\cE)$ and forward graph $G_F = (V, E_F)$. However, we do not have random access to the adjacency matrix of these graphs. Indeed, to establish if $(u, v) \in V_0 \times V_1$ is eligible we need to check the condition $\varphi(u) + \varphi(v) = c(u, v) + 1$ (or $\varphi(u) + \varphi(v) = c(u, v)$). However, we will see that the potential $\varphi(\cdot)$ requires more than a single query to be evaluated. Formally, we assume that we have a potential oracle $\eval_\varphi(\cdot)$ that returns the value of $\varphi(\cdot)$ in time $t_\varphi$. 
Whenever checking whether $(u, v)$ is an edge of $G_F$ ($G_\cE)$ requires to evaluate a condition of the form $\varphi(u) + \varphi(v) = c(u, v) + 1$ (or $\varphi(u) + \varphi(v) = c(u, v)$) we say that we have \emph{potential oracle access} to the adjacency matrix of $G_F$ ($G_\cE$) with potential oracle time $t_\varphi$.
We can think of $t_\varphi$ as $\Otil(n^{1+\varepsilon})$ and we will later prove that this is (roughly) the case.

\paragraph{Potential functions with constant-size range.}
If our potential function $\varphi: V \rightarrow \cR$ has range size $|\cR| \leq R$ then we say that it is an $R$-potential. If the eligibility (forward) graph is induced by $R$-potentials for $R=O(1)$ we can rephrase \Cref{thm:large matching} and \Cref{thm:augmenting paths} to work with potential oracle access, without any asymptotic overhead.
The following theorem is an analog of \Cref{thm:large matching} for forward graphs.

\begin{lemma}
\label{lem:large matching with potential}Let $G_F=(V,E_F)$ be a forward graph w.r.t the $R$-potential $\varphi$, let $A\subseteq V$ be a vertex set. Suppose we have a potential oracle $\eval_\varphi$ with oracle time $t_\varphi$ and an membership oracle $\mem_{A}$ with $t_{A}$ query time. We are given as input constants $0 < \varepsilon \leq 0.2$ and $\sizeIn>0$.

There exists an algorithm $\largematchforward(\varphi, A, \sizeIn)$ that preprocesses $G_F$ in $\Otil_R((t_{A} + t_\varphi+n) \cdot n^{1-\varepsilon})$   time and either returns $\bot$ or constructs a matching oracle $\match_{M}(\cdot)$ for a matching $M\subset G_F[A]$ of size at least $\sizeOut n$ where $\sizeOut=\frac{1}{2000 \cdot R^{10}}\sizeIn^{5} = \Theta_{\sizeIn, R}(1)$ that has $\Otil_R((t_{A} + t_\varphi+n)n^{4\varepsilon})$ 
worst-case query time. If $\mu(G_F[A])\ge\sizeIn n$, then $\bot$ is not returned. The guarantee holds with high probability.
\end{lemma}
\begin{proof}
Without loss of generality, we assume that $\varphi$ takes values in $[R]$.
Suppose that $G_F[A]$ has a matching of size $\sizeIn n$. We partition the edges $E_F[A] = E_F \cap (A\times A)$ into $R^2$ sets $E_{0, 0} \dots E_{R-1, R-1}$ such that $(u, v) \in E_{i, j}$ iff $\varphi(u) = i$ and $\varphi(v) = j$. Then, there exist $i, j \in [R]$ such that $G_{i, j} = (V, E_{i, j})$ has a matching of size $\sizeIn n/ R^2$.
Moreover, once we restrict ourselves to $G_{i, j}$, each edge query $(u, v) \in E_{i, j}$ becomes much easier. Indeed, we just need to establish if $i + j = c(u, v) + 1$. In order to restrict ourselves to $G_{i, j}$ it suffices to set $A' = A \cap (\varphi^{-1}(\{i\}) \times \varphi^{-1}(\{j\}))$. Then the membership oracle $\mem_{A'}$ runs in time $O(t_A + t_\varphi)$.
Hence, using \Cref{thm:large matching} we can find a matching of size $\sizeOut n$, where $\sizeOut=\frac{1}{2000 \cdot R^{10}}\sizeIn^{5}$. Algorithmically, we run the algorithm from \Cref{thm:large matching} $R^2$ times (once for each pair $(i, j)$) and halt as soon as the algorithm does not return $\bot$.
\end{proof}

The following is an analog of \Cref{thm:augmenting paths} for eligibility graphs. 
\begin{lemma}
\label{lem:augmenting paths with potential} 
Let $\varepsilon_\inp > 0$ be a sufficiently small constant.
Let $\alpha_{k, \gamma}$ and $\beta_{k, \gamma}$ be constants that depend on $k$ and $\gamma$ and set $\varepsilon_{\out} := \alpha_{k, \gamma} \cdot \varepsilon_{\inp}$. We have an $R$-potential oracle $\eval_\varphi$ with running time $t_\varphi = \Otil(n^{1+\varepsilon_{\inp}})$, a matching oracle $\match_{M^\inp}$ with running time $t_{M^\inp} = \Otil(n^{1+\varepsilon_{\inp}})$ and an eligibility graph $G_\cE = (V, E_\cE)$ w.r.t. $\varphi$ and $M^\inp$. 

There exists an algorithm $\augmenteligible(\varphi, M^{\inp}, k, \gamma, \varepsilon_{\inp})$ that runs in $\tilde{O}_{k, \gamma, R}\left(n^{2-\varepsilon_{\inp}}\right)$ time. Further, either it returns an oracle $\match_{M^{\out}}(\cdot)$ with query time $\tilde{O}_{k,\gamma, R}(n^{1+\varepsilon_{\out}})$, for some matching $M^{\out}$ in $G_\cE$ of size $|M^{\out}| \geq |M^{\inp}| + \beta_{k, \gamma} \cdot n$ (we say that it ``succeeds'' in this case), or it returns $\bot$.  Finally, if the matching $M^{\inp}$ admits a collection of $\gamma \cdot n$ many node-disjoint augmenting paths with length $\leq k$ in $G_\cE$, then the algorithm succeeds whp.
\end{lemma}
\begin{proof}
We derive \Cref{lem:augmenting paths with potential} combining \Cref{thm:augmenting paths} and \Cref{lem:large matching with potential}.
First, we notice that \Cref{thm:augmenting paths} says that the algorithm succeeds (whp) whenever there are $\gamma' n$ node-disjoint augmenting paths (NDAP) with length \emph{exactly} $2k'+1$, while \Cref{lem:augmenting paths with potential} has the weaker requirement that there are at least $\gamma n$ NDAP of length $\leq k$. A simple reduction is obtained invoking \Cref{thm:augmenting paths} with $\gamma' = \gamma / k$ for all $k'$ such that $2k' + 1 \leq k$ (notice that all augmenting paths have odd length). In this way, if there exists a collection of $\gamma n$ NDAP of length $\leq k$ then there exists a $k' \leq (k-1)/2$ such that we have $\gamma' n$ NDAP of length \emph{exactly} $2k' + 1$. All guarantees are preserved since we consider both $\gamma$ and $k$ constants.
Now, we are left to address the fact that we do not have random access to the adjacency matrix of $G_\cE$, but rather potential oracle access.

We notice that, according to \Cref{thm:augmenting paths}, the implementation of $\augment$ from \cite{bhattacharya2023dynamic} never makes any query to the adjacency matrix besides those performed inside $\largematch$. 
Moreover, \Cref{lem:large matching with potential} implies that $\largematchforward$ is not asymptotically slower than $\largematch$ as long as $R = O(1)$.
\end{proof}

Finally, we observe that in \Cref{alg:template algorithm} each potential is increased (or decreased) at most $T = O_{C, \gamma}(1)$ times. Hence, $\varphi$ is a $R$-potential for $R = 2T+1 = O_{C, \gamma}(1)$. Thus, we can consider $R$ a constant when applying \Cref{lem:large matching with potential} or \Cref{lem:augmenting paths with potential}.

\subsection{Implementing Step 1}
In this subsection we implement Step 1 from the template algorithm in sublinear time. Here we assume that we have at our disposal a potential oracle $\eval_{\varphi^\inp}$ running in time $t_{\varphi^\inp} = \Otil(n^{1+\varepsilon_{\inp}})$ and a matching oracle $\match_{M^\inp}$ with running time $t_{M^\inp} = \Otil(n^{1+\varepsilon_{\inp}})$.
We will output a potential oracle $\eval_{\varphi^\out}$ running in time $t_{\varphi^\out} = \Otil(n^{1+\varepsilon_{\out}})$ and a matching oracle $\match_{M^\out}$ with running time $t_{M^\out} = \Otil(n^{1+\varepsilon_{\out}})$.
We show that there exists a $(k, \xi)$-\qmsndap $\cA$ such that: the matching $M^\out$ is obtained from $M^{\inp}$ by augmenting it along all paths in $\cA$; $\varphi^{\out}$ is obtained from $\varphi^{\inp}$ by subtracting $1$ to $\varphi^{\inp}(v)$ for each $v \in V_1 \cap \bigcup_{P \in \cA} P$.

\begin{algorithm}
Set $k$ and $\gamma$ as in \Cref{alg:template algorithm}.

Initialize $M \leftarrow M^\inp$ and $\varepsilon \leftarrow \varepsilon_{\inp}$.

Repeat until $\augmenteligible$ returns $\bot$:
\begin{enumerate}
    \item Let $\augmenteligible(G_\cE, M, k, \gamma, \varepsilon)$ return $\match_{M'}$ with running time $t_{M'} = \Otil(n^{1+\varepsilon'})$.
    \item Update $M \leftarrow M'$ and $\varepsilon \leftarrow \varepsilon'$.
\end{enumerate}

Set $M^\out \leftarrow M$ and $\varepsilon_{\out} \leftarrow \varepsilon$.

\hrulefill

Implement $\eval_{\varphi^\out}(u)$ as follows:
\begin{itemize}
    \item If $u \in V_0$, return $\eval_{\varphi^\inp}$.
    \item Else, $u \in V_1$, set $ v \leftarrow \match_{M^\out}(u)$.
    \item If $v \texttt{ == } \bot$, return $\eval_{\varphi^\inp}(u)$.
    \item Else, we have $(u, v) \in M^\out$:
    \begin{itemize}
        \item If $c(u, v) + 1 \texttt{ == } \eval_{\varphi^\inp} (u) + \eval_{\varphi^\inp} (v)$, return $\eval_{\varphi^\inp}(u) - 1$.
        \item Else, return $\eval_{\varphi^\inp}(u)$. 
    \end{itemize}
\end{itemize}

\caption{Implementation of Step 1. \label{alg:implementation of step1}}
\end{algorithm}

\paragraph{Analysis.}
First, we observe that the algorithm above correctly implements the template, with high probability (all our statements henceforth hold whp). 
Initialize $\cA \leftarrow \emptyset$.
For each run of $\augmenteligible(G, M, k, \gamma, \varepsilon)$ we decompose $M \oplus M'$ into a set of augmenting paths $\cP$ and a set of alternating cycles $\cC$ and we set $\cA \leftarrow \cA \cup \cP$.
When $\augmenteligible(G, M, k, \gamma, \varepsilon)$ returns $\bot$ it means (by \Cref{lem:augmenting paths with potential}) that there are at most $\gamma n$ node-disjoint augmenting paths of length $\leq k$ that do not intersect $\bigcup \cA$. Hence, $\cA$ is a $(k, \xi)$-\qmsndap.
Clearly, $\match_{M^\out}$ implements the matching obtained from $M^\inp$ by augmenting along the paths in $\cA$.

To see that the implementation of $\eval_{\varphi^\out}(u)$ is correct it is sufficient to notice that in the template algorithm we decrement $\varphi(u)$ iff:
$(i)$ $u \in V_1$, and $(ii)$ there exists an augmenting path $P \in \cA$ intersecting $u$.
Since every node belongs to at most one path in $\cA$ then $u$ is matched in $M^{\out}$ and $(u, v) \in M^\out$ is an $M^\inp$-eligible edge. Thus, $(ii)$ is equivalent to: $(iii)$ $v = \match_{M^\out}$ satisfies $\varphi^\inp(v) + \varphi^\inp(u) = c(u, v) + 1$.  
Finally, we bound $\varepsilon_\out$ as a function of $\varepsilon_\inp$.

\begin{lemma}
Step 1 can be implemented in $\Otil(n^{2-\varepsilon})$ time for some constant $\varepsilon > 0$. 
Moreover, the oracle $\match_{M^\out}$ has running time $t_{M^\out}$ and the oracle $\eval_{\varphi^\out}$ has running time $t_{\varphi^\out}$ such that  $t_{M^\out}, t_{\varphi^\out} = \Otil(n^{1 + \varepsilon_\out})$ and $\varepsilon_\out = O_{\gamma, k}(\varepsilon_\inp)$.
\end{lemma}
\begin{proof}
Let $\varepsilon > 0$ and $\beta_{k, \gamma}$ as in \Cref{lem:augmenting paths with potential}.
\Cref{alg:implementation of step1} runs $\augmenteligible$ at most $1/\beta_{k, \gamma} + 1 = O_{k, \gamma}(1)$ times because the set $\cA$ increases by $\beta_{k, \gamma} \cdot n$ after each successful run of $\augmenteligible$.
Thus, there can be at most $1 / \beta_{k, \gamma}$ successful runs.
It is apparent that, by \Cref{lem:augmenting paths with potential}, Step 1 can be implemented in $\Otil_{k, \gamma}(n^{2-\varepsilon})$ time.

Now we prove the bound on oracles time.
First, we observe that $t_{\varphi^\out} = O(t_{M^\out} + t_{\varphi^\inp}) = \Otil(n^{1+\varepsilon_\out})$.
 Moreover, at every iteration we have $\varepsilon' \leq \alpha_{k, \gamma} \cdot \varepsilon$, hence $\varepsilon_{\out} \leq \alpha_{k, \gamma}^{1/\beta_{k, \gamma} + 1} \varepsilon_\inp = O_{k, \gamma}(\varepsilon_\inp)$.
\end{proof}

\subsection{Implementing Step 2}
In this subsection we implement Step 2 from \Cref{alg:template algorithm} in sublinear time. Once again, we assume that we have at our disposal a potential oracle $\eval_{\varphi^\inp}$ running in time $t_{\varphi^\inp} = \Otil(n^{1+\varepsilon_{\inp}})$ and a matching oracle $\match_{M^\inp}$ with running time $t_{M^\inp} = \Otil(n^{1+\varepsilon_{\inp}})$.
We will output a potential oracle $\eval_{\varphi^\out}$ running in time $t_{\varphi^\out} = \Otil(n^{1+\varepsilon_{\out}})$.
We show that there exists a $(k, \delta)$-\qmf $\cF$ with respect to $M^\inp$ such that 
$\varphi^\out(u) = \varphi^\inp(u) + 1$ for each $u \in \cF \cap V_0$ and $\varphi^\out(v) = \varphi^\inp(v) - 1$ for each $v \in \cF \cap V_1$.

The execution of \Cref{alg:implementation of step2} is represented in \Cref{fig:QMSNDAP and QMF}, where vertices colored in the same way are added to $\cF$ during the same iteration. 

\begin{algorithm}
Set $\delta$ as in \Cref{alg:template algorithm}.

Initialize $t \leftarrow 0$,  $\cF_0 \leftarrow F_0$, where $F_0$ is the set of $M^\inp$-unmatched vertices in $V_0$. 

Implement $\mem_{\cF_0}(u)$ as: $\match_{M^\inp}(u) \texttt{ == } \bot$. 

Repeat until $\largematchforward$ returns $\bot$:
\begin{enumerate}
    \item $A_t \leftarrow (\cF_t \cap V_0) \cup (V_1 \setminus \cF_t)$.
    \item Let $\largematchforward(\varphi, A_t, \delta)$ return $\match_{M_{t}}$.
    \item $\cF'_t \leftarrow \cF_t \cup \{\match_{M_{t}(u)} \cond u \in \cF_t\}$ 
    
    Implement $\mem_{\cF'_t}(u)$ as: $ \mem_{\cF_t}(u)$ or $\match_{M_{t}}(u) \in \cF_t$. \label{enum:first forest update}
    \item $\cF''_t \leftarrow \cF_t' \cup \{\match_{M^\inp}(u) \cond u \in \cF'_t\}$ 
    
    Implement $\mem_{F''_t}(u)$ as: $\mem_{\cF'_t}(u)$ or $\match_{M^\inp}(u) \in \cF'_t$. \label{enum:second forest update}
    \item $\cF_{t+1} \leftarrow \cF''_t$; $t \leftarrow t+1$.
\end{enumerate}
Implement $\eval_{\varphi^\out}(u)$ as  $\eval_{\varphi^\inp}(u) + \mem_{\cF_t}(u) \cdot (-1)^{\one_{u \in V_1}}$
\caption{Implementation of Step 2. \label{alg:implementation of step2}}
\end{algorithm}

\paragraph{Analysis.}
First, we prove that \Cref{alg:implementation of step2} implements the template (all guarantees hold whp). Namely, that $\cF$ is a $(k, \delta)$-\qmf, where $k$ and $\delta$ are defined as in \Cref{alg:template algorithm}. With a slight abuse of notation, in \Cref{alg:implementation of step2} we used $\cF$ to denote the set of nodes in the forest. Here, we understand that for each $u \in \cF'_{t} \setminus \cF_t$ we have an edge $(u, \match_{M_{t+1}}(u))$ and for each $v \in \cF''_{t} \setminus \cF'_t$ we have an edge $(v, \match_{M^\inp}(v))$. 
Let $\tau$ be the total number of times $\largematchforward$ runs successfully in \Cref{alg:implementation of step2}. We will see that $\tau \leq k / 2$. Notice that $\cF_\tau$ is the last forest produced by \Cref{alg:implementation of step2} and for each $u \in \cF_\tau \setminus F_0$ we add an edge incident to $u$, thus $\cF_\tau$ is a forest with $|F_0|$ connected components, one for each $u \in F_0$. 
Now we show that $\cF_\tau$ is a $(k, \delta)$-\qmf w.r.t. $M^\inp$. We refer to the notation of \Cref{def:qmf}. \Cref{enum:forest def free included} is clearly satisfied. \Cref{enum:forest def matched edges included} is satisfied because of line $\ref{enum:second forest update}$ in \Cref{alg:implementation of step2}.

Now we show that \Cref{enum:forest def forest depth} is satisfied. Define $k =6000 (2T+1)^{10} / \delta^5 $ as in \Cref{alg:template algorithm} and recall that $\varphi$ is a $(2T+1)$-potential.
Thanks to \Cref{lem:large matching with potential}, at each step we increment $|\cF|$ by at least $\frac{1}{2000 (2T+1)^{10}}\cdot \delta^5 n$. Thus, no more than $\lceil 2000 (2T+1)^{10} / \delta^5 \rceil \leq k /2$ iterations are performed and we cannot have more than $k$ hops between $u \in \cF$ and $v \in F_0$ if $u$ belongs to the connected component of $v$.

Now we prove that \Cref{enum:forset def cc size} is satisfied. At each iteration, the size of each connected component of $\cF_t$ at most triples. Indeed, let $C$ be a connected component of $\cF$. In step \ref{enum:first forest update} we add to $C$ at most $|C|$ vertices (because we add a vertex for each edge in a matching incident to $C$) and in step \ref{enum:second forest update} we add to $C$ at most one more vertex for each new vertex added in step \ref{enum:first forest update}.

Now we prove that \Cref{enum:forest def vertex cover} is satisfied. \Cref{alg:implementation of step2} halts when $\largematchforward(M^\inp, (\cF \cap V_0) \cup (V_1 \setminus \cF), \delta)$ returns $\bot$. This may only happen when there is no matching between of $\cF \cap V_0$ and $V_1 \setminus \cF$ of size $\delta n$. This implies that there exists a vertex cover of size $\leq \delta n$. Moreover, this is a vertex cover for the whole $E_\cE \cap (\cF \times V \setminus \cF)$ because all edges in $E_\cE \cap V_1 \times V_0$ are in $M^\inp$ and by \Cref{enum:forest def matched edges included} have both endpoints either in $\cF$ or in $V\setminus \cF$.

It is easy to check that $\varphi^\out(u) = \varphi^\inp(u) + 1$ for each $u \in \cF \cap V_0$ and $\varphi^\out(v) = \varphi^\inp(v) - 1$ for each $v \in \cF \cap V_1$.

\begin{lemma}
Step 2 can be implemented in $\Otil(n^{2-\varepsilon})$ time for some constant $\varepsilon > 0$. 
Moreover, the oracle $\eval_{\varphi_\out}$ has running time $t_{\varphi^\out} = \Otil(n ^{1+ \varepsilon_\out})$ and $\varepsilon_\out = O_{k, \delta}(\varepsilon_\inp)$.
\end{lemma}
\begin{proof}
For $s = 0 \dots \tau$ denote with $\varepsilon_s > 0$ a constant such that $t_{A_s} = \Otil(n^{1+\varepsilon_s})$, where $t_{A_s}$ is the running time of $\mem_{A_s}$.
Notice that $\varepsilon_0 = \varepsilon_\inp$.
At step $s$ we choose $\hat \varepsilon_s := 2 \varepsilon_s$ as the $\varepsilon$ parameter in \Cref{lem:large matching with potential}. This implies that $\largematchforward$ runs in $\Otil(n^{1+\varepsilon_s} \cdot n^{1- \hat \varepsilon_s}) = \Otil(n^{2 - \varepsilon_s})$ time.
We have already proved that $ \tau \leq k$, thus \Cref{alg:implementation of step1} takes $\Otil(n^{2-\varepsilon})$ time in total, where $\varepsilon := \min_{s \in [0, \tau]} \varepsilon_s = \varepsilon_\inp$.

Denote with $t_\cF$ the query time of $\mem_\cF$.
For each $s$, we have $t_{\cF_{s+1}} = t_{M_{s+1}} + t_{M^\inp} + t_{\cF_{s}}$.
Thanks to \Cref{lem:large matching with potential} we have $t_{M_{s+1}} = \Otil((t_{A_s} + t_{\varphi^\inp} + n) n^{4\hat \varepsilon_s})$. Moreover, $t_{A_s} = t_{\cF_s}$. Thus, $t_{\cF_{s+1}} = (t_{\cF_s} + t_{\varphi^\inp} + n) n^{8 \varepsilon_s} + t_{M^\inp} + t_{\cF_{s}}$.
Since, $t_{\cF_0} = t_{\varphi^\inp} = t_{M^\inp} = \Otil(n^{1+\varepsilon_\inp})$ we have $t_{\varphi^\out} = t_{\varphi^\inp} + t_{\cF_\tau} = \Otil(n^{1+ 9^\tau \cdot \varepsilon_\inp}) = \Otil_{k}(n^{1+ O_{k, \delta}(\varepsilon_\inp)})$.
\end{proof}

\subsection{Implementing the Template Algorithm}

We can put together the results proved in the previous subsections and show that \Cref{alg:template algorithm} can be implemented in sublinear time.

\begin{theorem}
\label{thm:template algo in sublinear time}
There exists a constant  $\varepsilon > 0$ such that \Cref{alg:template algorithm} can be implemented in time $O(n^{2-\varepsilon})$.
Moreover, using the notation in \Cref{obs:how to build matching oracle}, we can return a matching oracle $\match_{M^\alg_{[\alpha_w]}}$ running in time $O(n^{1-\varepsilon})$ such that $M^\alg_{[\alpha_w]}$ satisfies $|M^\alg_{[\alpha_w]}| \geq (1-\gamma) n$ and $c(M^\alg_{[\alpha_w]}) \leq c(M^\opt)$ and $\match_{M^\alg_{[\alpha_w]}}$.
\end{theorem}

\begin{proof}
\Cref{alg:template algorithm} runs $T = O_{C, \gamma}(1)$ iterations, and a single iterations consists of Step 1 and Step 2. 
At iteration $s$ denote with $\varepsilon^{(s)}_\inp$ the value of $\varepsilon_\inp$ for Step 1 input (or, equivalently, the value of $\varepsilon_\out$ for Step 2 output at iteration $s-1$) and with $\varepsilon^{(s)}_\out$ the value of $\varepsilon_\out$ for Step 1 output (or, equivalently, the value of $\varepsilon_\inp$ for Step 2 input at iteration $s$).
Every time we run either Step 1 or Step 2, the value of $\varepsilon_\out$ is at most some constant factor larger than  $\varepsilon_\inp$.
This translates into $\varepsilon^{(s)}_\inp = O_{C, k, \gamma}(\varepsilon^{(s-1)}_\out)$ and $\varepsilon^{(s)}_\out = O_{C, k, \gamma}(\varepsilon^{(s)}_\inp)$.
Thus, after $T$ iterations $\varepsilon^{(T)}_\out$ is arbitrarily small, provided that $\varepsilon^{(0)}_\inp$ is small enough.
To conclude, we notice that the initial matching is empty and the initial potential is identically $0$, so the first membership oracle and potential oracle run in linear linear, thus we can set $\varepsilon^{(0)}_\inp$ arbitrarily small.
Finally, let $M^\alg$ be the last matching computed by \Cref{alg:template algorithm}. We have at our disposal a matching oracle $\match_{M^\alg}$ running in time $\Otil(n^{1+\varepsilon_\inp^{(T+1)}})$, so we can easily sample the a $S$ of $O(\log n)$ edges from $M^\alg$ in time $\Otil(n^{1+\varepsilon_\out^{(T)}})$. This conclude the implementation of \Cref{alg:template algorithm}.

Moreover, we compute $w$ as the largest value such that at least $3\gamma|S|$ edges in $S$ have cost $\geq w$ and define $\alpha_w$ such that exactly $\alpha_w \cdot n$ edges in $M^\alg$ have cost $\geq w$. Then, we implement a matching oracle $\match_{M^\alg_{[\alpha_w]}}$ for $M^\alg_{[\alpha_w]}$ running in time $\Otil(n^{1+\varepsilon_\out^{(T)}})$ as follows: given $u \in V_0 \cup V_1$ we set $v \leftarrow \match_{M^\alg}(u)$; if $c(u, v) < w$ then we return $v$, else we return $\bot$. Thanks to \Cref{obs:how to build matching oracle}, we have $|M^\alg_{[\alpha_w]}| \geq (1-\gamma) n$ and $c(M^\alg_{[\alpha_w]}) \leq c(M^\opt)$.
\end{proof}
\section{Proof of our Main Theorems}
\label{sec:general costs and emd}

In this section we piece things together and prove \Cref{thm:main theorem mwm}. Then, we use \Cref{thm:main theorem mwm} to prove \Cref{thm:main theorem intro EMD}, \Cref{thm:main theorem intro EMD generalized} and \Cref{thm:main theorem intro graph}.

\subsection{Proof of \Cref{thm:main theorem mwm}}
In this subsection we strengthen \Cref{thm:weaker main thm mwm}, extend its scope to arbitrary costs and combine it with \Cref{thm:template algo in sublinear time} to obtain \Cref{thm:main theorem mwm}. We restate the latter for convenience. 

\MainThmMWM*

\paragraph{Roadmap of the proof.}
\Cref{thm:weaker main thm mwm} works only for weights in $[1, C]$. In order to reduce to that case, we need to find a \emph{characteristic cost} $\wbar$ of min-weight matchings with size in $[\alpha n, \beta n]$. Then, we round every cost to a multiple of $\frac{1}{2}\gamma^2 \wbar$, where $\gamma$ is a small constant. We show that, thanks to certain properties of the characteristic cost $\wbar$, the approximation error induced by rounding the costs is negligible. 
Finally, we pad each size of the bipartition with dummy vertices to reduce the problem of finding a matching of approximate size $\beta n$ to that of finding an \emph{approximate perfect matching}, which is addressed in \Cref{thm:weaker main thm mwm}.

\paragraph{Notation.}
Similarly to \Cref{thm:main theorem mwm}, we denote with $M^\xi$ the min-weight matching of size $\xi n$ in $G$. Likewise, we will define a graph $\Gbar$ and denote with $\Mbar^\xi$ the min-weight matching of size $\xi n$ in $\Gbar$.
As in \Cref{sec:implementing in sublinear time}, given a matching $M$, we denote with $M_{[\delta]}$ the matching obtained from $M$ by removing the $\delta n$ most expensive edges.
We denote with $\mu(M)$ the cost of the most expensive edge in $M$. 
Given $w \geq 0$, we denote with $G_{\leq w}$ the graph of edges which cost $\leq w$.
Throughout this subsection, fix a constant $0 < \gamma < (\beta - \alpha) / 4$.

\paragraph{Reduction from arbitrary weights to $[1, C]$.}
The next technical lemma shows that, if we can can solve an easier version of the problem in \Cref{thm:main theorem mwm} where we allow an additive error $\gamma^2 \wbar n$ on an instance where $\wbar$ is an upper bound for the cost function; then we can also solve the problem in \Cref{thm:main theorem mwm}. This reduction is achieved by finding a suitable characteristic cost $\wbar$ in sublinear time and running the aforementioned algorithm on $G_{\leq \wbar}$.

\begin{lemma}
\label{lem:find w bar}
Suppose that there exists an algorithm that takes as input a bipartite graph $\Gbar=(\Vbar_0 \cup \Vbar_1, \Ebar)$ endowed with a cost function $\cbar: \Ebar \rightarrow [0, \wbar]$, outputs an estimate $\hat c$ and a matching oracle $\match_{\hat M}$ such that (whp) $\hat c$ satisfies
\begin{align*}
\label{eq:additive approximation}
\cbar(\Mbar^\alpha) \leq \hat c\leq \cbar(\Mbar^{\alpha + \gamma}) + \gamma^2 \wbar n 
\end{align*}
while $\hat M$ satisfies $|\hat M| \geq \alpha n$ and $c(\hat M) \leq \cbar(\Mbar^{\alpha + \gamma}) + \gamma^2 \wbar n$.
Suppose also that such algorithm runs in time $O(\bar n^{2-\varepsilon})$ and $\match_{\hat M}$ runs in time $O(\bar n^{1+\varepsilon})$ for some $\varepsilon > 0$, where $\bar n = |\Vbar_0|=|\Vbar_1|$.

Then, there exists an algorithm that takes as input a bipartite graph $G=(V_0 \cup V_1, E)$ endowed with a cost function $c: E \rightarrow \bbR^+$, outputs an estimate $\hat c$ and a matching oracle $\match_{\hat M}$ such that (whp) $\hat c$ satisfies
\[
C(M^\alpha) \leq \hat c \leq c(M^\beta) 
\]
while $\hat M$ satisfies $|\hat M| \geq \alpha n$ and $c(\hat M) \leq c(M^\beta)$.
Moreover, such algorithm runs in time $O(n^{2-\varepsilon})$ and $\match_{\hat M}$ runs in time $O(n^{1+\varepsilon})$ for some $\varepsilon > 0$, where $n = |V_0|=|V_1|$.
\end{lemma}
\begin{proof}
First, we show how to compute, in time $O(n^{2-\varepsilon})$, a value $\wbar$ such that:
\begin{enumerate}[(i)]
    \item $\gamma \cdot \wbar \leq \mu(M^\beta_{[\gamma]})$
    \item $\wbar \geq \mu(M^{\alpha + 3 \gamma}_{[2\gamma]})$.
\end{enumerate}

We sample $s = \Theta(n \log n / \gamma)$ edges from $G$ uniformly at random. Let $w_1 \leq w_2 \leq \dots \leq w_s$ be their costs. 
Recall that \Cref{thm:mcm in sublinear} allows us to compute the size of a maximal-cardinality matching (MCM) of the graph $G$, up to a $\gamma n$-additive approximation, in time $O(n^{2-\Omega_\gamma(1)})$. Denote that algorithm with $\approxmatch(G, \gamma)$.
Using binary search, we find the largest cost $w_i$ such that $\approxmatch(G_{\leq \gamma w_i}, \gamma)$ returns an estimated MCM size $< (\beta -  2 \gamma) n$.
Then, we set $\wbar := w_i$. 

We prove that property $(i)$ holds. Suppose that $\mu(M^\beta_{[\gamma]}) < \gamma w_i$. Then, in $G_{\leq \gamma w_i}$ there exists a matching of size $(\beta - \gamma) n$, therefore $\approxmatch(G_{\leq \gamma w_i}, \gamma)$ finds a matching of size $\geq (\beta - 2\gamma)n$ whp, contradiction. 

We prove that property $(ii)$ holds.
First, we prove that $w_{i+1} \geq \mu(M^{\alpha + 3\gamma}_{[\gamma]})$. Indeed, suppose the reverse (strict) inequality holds. Then, for each matching $M$ of size $(\beta - 2\gamma) n$ we have
\[
c(M) \geq c(M^{\beta - 2 \gamma}) \geq c(M^{\alpha + 3\gamma}) > \gamma \cdot w_{i+1} n
\]
which implies that there exists $e \in M$ such that $c(e) > \gamma \cdot w_{i+1}$. However, this cannot hold for each matching $M$ of size $(\beta -2\gamma)n$ because $\approxmatch(G_{\leq \gamma w_{i+1}}, \gamma)$ returned (whp) a matching (oracle) of size $ \geq (\beta - 2\gamma) n$. Contradiction. 
Therefore, we have $w_{i+1} \geq \mu(M^{\alpha + 3\gamma}_{[\gamma]})$. However, since we sampled $\Theta(n \log n / \gamma)$ edges, then for each $i$ there are, whp, at most $\gamma \cdot n$ edges which cost $w$ satisfies $w_i < w < w_{i+1}$. Hence, removing $\gamma n$ more edges from $M^{\alpha + 3\gamma}_{[\gamma]}$ we obtain $w_i \geq \mu(M^{\alpha + 3\gamma}_{[2\gamma]})$. 

We define $\Gbar := G_{\leq \wbar}$ and $\cbar = c|_{\Ebar}$, run the algorithm in the premise of the lemma on $\Gbar$ and $\cbar$ and let $\hat c$, $\match_{\hat M}$ be its outputs.
It is apparent that this reduction takes $O(n^{2-\varepsilon})$ for some $\varepsilon > 0$.

Since $\Gbar$ is a subgraph of $G$, we have $c(M^\alpha) \leq c(\Mbar^\alpha)$.
Moreover, conditions $(i)$ and $(ii)$, together with $5\gamma \leq \beta - \alpha $ imply
\[
c(\Mbar^{\alpha + \gamma}) \leq c(M^{\alpha + 3\gamma}_{[2\gamma]}) \leq c(M^{\alpha + 3\gamma}) \leq c(M^{\beta - \gamma}) \leq
c(M^\beta_{[\gamma]}) \leq c(M^\beta) - \gamma n \cdot \mu(M^\beta_{[\gamma]}) \leq c(M^\beta) - \gamma^2 \wbar n.
\]
Thus, $c(\Mbar^{\alpha}) \leq \hat c \leq c(\Mbar^{\alpha + \gamma}) + \gamma^2 \wbar n $ implies $c(M^\alpha) \leq \hat c \leq c(M^\beta)$ and $c(\hat M) \leq  c(\Mbar^{\alpha + \gamma}) + \gamma^2 \wbar n $ implies $\hat c \leq c(M^\beta)$.
\end{proof}

The following lemma shows how to reduce from real-values costs in $[0, w]$ (where possibly $w = \omega(1)$) to the more tame case where costs are integers in $[1, C]$. This reduction is achieved via rounding.   

\begin{lemma}
\label{lem:cost rounding}
Suppose that there exists an algorithm that takes as input a bipartite graph $\Gbar=(\Vbar_0 \cup \Vbar_1, \Ebar)$ endowed a cost function $\cbar: E \rightarrow [1, C]$, with $C=O(1)$, returns an estimate $\hat c$ and a matching oracle $\match_{\hat M}$ such that (whp) $\hat c$ satisfies
\[
\cbar(\Mbar^\alpha) \leq \hat c \leq \cbar(\Mbar^\beta)
\]
while $\hat M$ satisfies $|\hat M| \geq \alpha n$ and $\bar c(\hat M) \leq \cbar(\Mbar^\beta)$.
Suppose also that such algorithm runs in time $O(\bar n^{2-\varepsilon})$ and $\match_{\hat M}$ runs in time $O(\bar n^{1+\varepsilon})$ for some $\varepsilon > 0$, where $\bar n = |\Vbar_0|=|\Vbar_1|$.

Then, there exists an algorithm that takes as input a bipartite graph $G=(V_0 \cup V_1, E)$ endowed a cost function $c: E \rightarrow [0, w]$ (possibly $w = \omega(1)$), returns an estimate $\tilde c$ and a matching oracle $\match_{\tilde M}$ such that (whp) $\tilde c$ satisfies
\[
c(M^\alpha) \leq \tilde c \leq c(M^\beta) + \gamma^2 w n
\]
while $\tilde M$ satisfies $|\tilde M| \geq \alpha n$ and $c(\tilde M) \leq c(M^\beta) + \gamma^2 w n$.
Moreover, such algorithm runs in time $O(n^{2-\varepsilon})$ and $\match_{\hat M}$ runs in time $O(n^{1+\varepsilon})$ for some $\varepsilon > 0$, where $n = |V_0|=|V_1|$.
\end{lemma}
\begin{proof}
We define $\cbar(e) = \ld\lceil \frac{2 c(e)}{\gamma^2 w}\rd\rceil + 1$. Then, the maximum value of $\cbar$ on $\Gbar$ is $C:= 2/\gamma^2 + 2 = O(1)$.
We set $\Gbar = G$ and run the algorithm in the premise of the lemma on $\Gbar$ and $\cbar$. Let $\hat c$, $\match_{\hat M}$ be its outputs. We define $\tilde c := \frac{1}{2} \gamma^2 w\cdot \hat c$ and $\tilde M := \hat M$.
The definition of $\cbar$ implies that, for each edge $e$ in $\Gbar$, $c(e) \leq \frac{1}{2} \gamma^2 w \cdot \cbar(e) \leq c(e) + \gamma^2 w$.
Hence,
\[
\frac{1}{2} \gamma^2 w \cdot \cbar(\Mbar^\alpha) \geq c(\Mbar^\alpha) \geq c(M^\alpha)
\]
and
\[
\frac{1}{2} \gamma^2 w \cdot  \cbar(\Mbar^\beta) \leq \frac{1}{2} \gamma^2 w \cdot  \cbar(M^\beta) \leq 
c(M^\beta) + \gamma^2 w n.
\]
Therefore, $\cbar(\Mbar^\alpha) \leq \hat c \leq \cbar(\Mbar^\beta)$ implies $c(M^\alpha) \leq \tilde c \leq c(M^\beta) + \gamma^2 w n$
and $\bar c(\hat M) \leq \cbar(\Mbar^\beta)$ implies $c(\hat M) \leq \frac{1}{2} \gamma^2 w \cdot \bar c(\hat M) \leq \frac{1}{2} \gamma^2 w \cdot \bar c(\Mbar^\beta) \leq c(M^\beta) + \gamma^2 w n$.
\end{proof}

\paragraph{Reduction from size-$\beta n$ matching to perfect matching.}
The following lemma shows that, if we can approximate the min-weight of a perfect matching (allowing $\delta n$ outliers), then we can approximate the min-weight of a size-$\beta n$ matching  (allowing $(\beta - \alpha)n$ outliers). This reduction is achieved by padding the original graph with dummy vertices.

\begin{lemma}
\label{lem:dummy vertices}
Suppose that, for each $\delta > 0$, there exists an algorithm that takes as input a bipartite graph $\Gbar=(\Vbar_0 \cup \Vbar_1, \Ebar)$ endowed a cost function $\cbar: E \rightarrow [1, C]$, with $C=O(1)$, returns an estimate $\hat c$ and a matching oracle $\match_{\hat M}$ such that (whp) $\hat c$ satisfies
\[
\cbar(\Mbar^{1- \delta}) \leq \hat c \leq \cbar(\Mbar^1)
\]
while $\hat M$ satisfies $|\hat M| \geq (1-\delta) n$ and $c(\hat M)  \leq \cbar(\Mbar^1)$.
Suppose also that such algorithm runs in time $O(\bar n^{2-\varepsilon})$ and $\match_{\hat M}$ runs in time $O(\bar n^{1+\varepsilon})$ for some $\varepsilon > 0$, where $\bar n = |\Vbar_0|=|\Vbar_1|$.

Then, for each $0 \leq \alpha < \beta \leq 1$ there exists an algorithm that takes as input a bipartite graph $G=(V_0 \cup V_1, E)$ and $c:E\rightarrow [1, C]$, returns an estimate $\tilde c$ and a matching oracle $\match_{\tilde M}$ such that (whp) $\tilde c$ satisfies
\[
c(M^\alpha) \leq \tilde c \leq c(M^\beta)
\]
while $\tilde M$ satisfies $|\tilde M| \geq \alpha n$ and $c(\tilde M) \leq c(M^\beta)$.
Moreover, such algorithm runs in time $O(n^{2-\varepsilon})$ and $\match_{\tilde M}$ runs in time $O(n^{1+\varepsilon})$ for some $\varepsilon > 0$, where $n = |V_0|=|V_1|$.
\end{lemma}
\begin{proof}
Fix a constant $0 < \xi \leq (\alpha - \beta) / 2$. We construct $\Gbar$, starting from $\Gbar = G$ and $\cbar = c$, as follows.
We add a set of $(1-\beta + \xi) n$ dummy vertices on each side of the bipartition: $\Vbar_0 = V_0 \cup D_0$ and $\Vbar_1 = V_1 \cup D_1$. Add an edge $(d_0, v_1)$ to $E$ for each $(d_0, v_1) \in D_0 \times V_1$ and set $\cbar(d_0, v_1) = 1$. Do the same for each $(v_0, d_1) \in V_0 \times D_1$. Notice that we construct both the adjacency matrix and the cost function implicitly, because an explicit construction would take $\Omega(n^2)$ time.
We have $\bar n := |\Vbar_0| = |\Vbar_1| = (2-\beta + \xi)n$.

Set $\delta = \frac{\xi n}{2\bar n} $. Run the algorithm in hypothesis on $\Gbar$ and let $\hat c$ and $\match_{\hat M}$ be its outputs.
We set $\tilde c = \hat c - (2- 2\beta + \xi) n$. We set $\tilde M := \hat M \cap (V_0 \times V_1)$ and implement $\match_{\tilde M}(u)$ as follows. Let $v \leftarrow \match_{\hat M}(u)$. If $v = \bot$, return $\bot$. If either $u$ or $v$ is dummy, return $\bot$; else return $v$. 
It is easy to see that $\Mbar^1 \cap (V_0 \times V_1)$ is a min-weight matching of size $(\beta - \xi) n$ in $G$, hence 
\[
\cbar(\Mbar^1) = c(M^{\beta - \xi}) + |D_0| + |D_1| = c(M^{\beta - \xi}) + 2 (1-\beta + \xi)n \leq c(M^\beta) + (2-2\beta + \xi) n.
\]
On the other hand, in $\Mbar^{1-\delta}$ at most $2 \delta \bar n = \xi n$ dummy vertices are left unmathced, so at least $(2-2\beta + \xi)n$ dummy vertices are matched in $\Mbar^{1-\delta}$. Moreover, $\Mbar^{1-\delta} \cap (V_0 \times V_1)$ is a matching in $G$ of size $ \geq n - (1-\beta +\xi) n - \delta \bar n = (\beta - 2\xi) n$. Hence, 
\[
\cbar(\Mbar^{1-\delta}) \geq c(M^{\beta -2\xi}) + (2-2\beta+\xi) n \geq c(M^\alpha) + (2-2\beta +\xi) n.
\]
Thus, $\cbar(\Mbar^{1-\delta}) \leq \hat c \leq \cbar(\Mbar^1)$ implies $c(M^\alpha) \leq \tilde c \leq c(M^\beta)$.

Now we prove the bounds on $\tilde M$. We have that since $|\hat M| \geq (1-\delta) \bar n$, then at most $2\delta \bar n = \xi n$ dummy vertices are left unmatched in $\hat M$ and so
\begin{align*}
c(\tilde M) &\leq \cbar(\hat M) - (2(1-\beta + \xi)n - 2\delta \bar n) \\
&\leq \cbar(\Mbar^1) -(2 -2\beta +\xi) n \\
&= c(M^{\beta - \xi}) + 2(1-\beta -\xi)n -(2 -2\beta +\xi) n \\ 
&= c(M^{\beta - \xi}) + \xi n \leq c(M^\beta).  
\end{align*}
Moreover, $|\tilde M| \geq (1-\delta)\bar n -2(1-\beta + \xi) n = (\beta - \frac{3}{2}\xi)n \geq \alpha n$.
\end{proof}

Finally, we can prove \Cref{thm:main theorem mwm}.

\begin{proof}[Proof of \Cref{thm:main theorem mwm}]
We notice that combining \Cref{thm:weaker main thm mwm} and \Cref{thm:template algo in sublinear time} we have a sublinear implementation of \Cref{alg:template algorithm} that takes a graph bipartite graph $G=(V_0 \cup V_1, E)$ and a cost function $c:E \rightarrow [1, C]$ as input, outputs an estimate $\hat c$ and a matching oracle $\match_{\hat M}$. The estimate $\hat c$ satisfies $c(M^{1-\delta}) \leq \hat c \leq c(M^1)$, $\hat M$ satisfies $|\hat M| \geq (1-\delta) n$ and $c(\hat M) \leq c(M^\opt)$. Moreover, such algorithm runs in time $O(n^{2-\varepsilon})$ and $\match_{\hat M}$ runs in time $O(n^{1+\varepsilon})$ for some $\varepsilon > 0$.

Then, combining \Cref{lem:dummy vertices}, \Cref{lem:cost rounding} and \Cref{lem:find w bar} we obtain an algorithm that takes as input a bipartite graph $G=(V_0 \cup V_1, E)$ endowed with a cost function $c:E\rightarrow \bbR^+$, outputs an estimate $\hat c$ and a matching oracle $\match_{\hat M}$ such that (whp) $c(M^\alpha) \leq \hat{c} \leq c(M^\beta)$, $|\hat M| \geq \alpha n$, and $c(\hat M) \leq c(M^\beta)$. Moreover such algorithm runs in time $O(n^{2-\varepsilon})$ and $\match_{\hat M}$ runs in time $O(n^{1+\varepsilon})$ for some $\varepsilon > 0$.
\end{proof}

\subsection{Proof of \Cref{thm:main theorem intro EMD} and \Cref{thm:main theorem intro EMD generalized}}

Since \Cref{thm:main theorem intro EMD generalized} is more general than \Cref{thm:main theorem intro EMD} we simply prove the former.

\MainThmIntroEMDGen*

Fix a constant $\gamma > 0$.
From each probability distribution $\mu, \nu$ we sample (with replacement) a multi-set of $m = \Theta(n \log(n))$ points. We use $V_{\mu}, V_{\nu}$ to denote the respective multi-sets, and $\hat{\mu}, \hat{\nu}$ to denote the empirical distributions of sampling a random point from $V_{\mu}, V_{\nu}$.
Let $\cT_\mu, \cT_\nu$ be the transport plans realizing $\emd(\mu, \mu')$ and $\emd(\nu, \nu')$ respectively. Namely, $\cT_\mu$ is a coupling between $\mu$ and $\mu'$ such that $\emd(\mu, \mu') = E_{(x, y)\sim \cT_\mu}[d(x, y)]$ and likewise for $\cT_\nu$.
For each sample $x$ in $\hat \mu$ we sample $x' \sim \cT(x, \cdot)$ and let $V'_\mu$ be the multi-set of samples $x'$ for $x \in V_\mu$. Define $V'_\nu$ similarly. Let $\hat \mu'$ and $\hat \nu'$ be the empirical distributions of sampling a random point from $V'_\mu$ and $V'_\nu$. 

\begin{lemma}
\label{lem:empirical are close}
$\emd(\hat \mu, \hat \mu') \leq \xi + \gamma$ with high probability.
\end{lemma}
\begin{proof}
$\E[\emd(\hat \mu, \hat \mu')] \leq \E\ld[\frac{1}{m} \cdot \sum_{x \in V_\mu} d(x, x')\rd] = \E_{(x, x') \sim \cT_\mu}[d(x, x')] = \emd(\mu, \mu') \leq \xi$. Moreover, $\Var_{(x, x') \sim \cT_\mu}[d(x, x')] \leq 1$, thus $m = \Theta(n\log n)$ ensures $\emd(\hat \mu, \hat \mu') \leq \xi + \gamma$ whp.
\end{proof}

\begin{lemma}
\label{lem:empirical close to discretized}
$\emd(\hat \mu', \mu') \leq \tv(\hat \mu', \mu') \leq \gamma$ with high probability.
\end{lemma}
\begin{proof}
First, we observe that $V_\mu'$ is distributed as a multi-set of $m$ samples from $\mu'$.
For any point $x'$ with at least $(\gamma/4n)$-mass in $\mu'$, we expect $\Omega(\log n)$ samples of $x'$ in $V'_\mu$, so by Chernoff bound we have that with high probability the number of samples of $x'$ concentrates to within $(1 \pm \gamma/4)$-factor of its expectation. Furthermore, with high probability at most $(\gamma/2)$-fraction of the samples correspond to points with less than $(\gamma/4n) $-mass in the original distribution. Thus overall, the empirical distribution $\hat{\mu}'
$ is within $\gamma$ TV distance of $\mu'$.
Finally, $\emd(\hat \mu', \mu') \leq \tv(\hat \mu', \mu')$ beacuse $d(\cdot, \cdot) \in [0, 1]$.
\end{proof}

\begin{lemma}
\label{lem:putting things together}
$\emd(\mu, \hat \mu) \leq \xi + 2\gamma$ with high probability.
\end{lemma}
\begin{proof}
Combining \Cref{lem:empirical are close}, \Cref{lem:empirical close to discretized} we obtain the following
\[
\emd(\mu, \hat \mu) \leq \emd(\mu, \mu') + \emd(\mu', \hat \mu') + \emd(\hat \mu', \hat \mu) \leq  2\xi + 2\gamma.
\]
\end{proof}

\begin{proof}[Proof of \Cref{thm:main theorem intro EMD generalized}]

We consider the bipartite graph with a vertex for each point in $V_{\mu}, V_{\nu}$ and edge costs induced by $d(\cdot,\cdot)$. We apply the algorithm guaranteed by \Cref{thm:main theorem mwm} to find an estimate of the min-weight matching over between $(1-\gamma) m$ and $m$ vertices. 
We return the cost estimate $\widehat{\emd}$ on the bipartite graph (normalized by dividing by $m$). 
\Cref{thm:main theorem mwm} guarantees that $\widehat{\emd} \in [\emd(\hat \mu, \hat \nu) \pm \gamma]$. Then using triangle inequality on $\emd$, as well as \Cref{lem:putting things together} on both $\mu$ and $\nu$ we obtain

\[
\ld|\widehat{\emd} - \emd(\mu, \nu)\rd| \leq \ld|\widehat{\emd} - \emd(\hat \mu, \hat \nu)\rd| + \emd(\mu, \hat \mu) + \emd(\nu, \hat \nu) \leq
4\xi + 5\gamma.
\]
Scaling $\gamma$ down of a factor $5$ we retireve \Cref{thm:main theorem intro EMD generalized}.
\end{proof}

\subsection{Proof of \Cref{thm:main theorem intro graph}}

\MainThmIntroGraph*

\begin{proof}
Let $M$ be the maximum matching in $G$ with total cost at most $B$, and let $|M| = \xi n$.
We perform a binary search for $\xi$ using the algorithm from \Cref{thm:main theorem mwm} as a subroutine. This loses only a factor $\log(n)$ in query complexity, which gets absorbed in $O(n^{2-\varepsilon})$ by choosing a suitable constant $\varepsilon$.
\end{proof}

\paragraph{Acknowledgments.}
We thank Erik Waingarten for inspiring discussions.
We thank Tal Herman and Greg Valiant for pointing out the sample complexity lower bound implied by \cite{valiant2010clt}. 

\printbibliography
\end{document}